%% file: kdd.tex
\documentclass[sigconf]{acmart}

\AtBeginDocument{%
  \providecommand\BibTeX{{%
      \normalfont B\kern-0.5em{\scshape i\kern-0.25em b}\kern-0.8em\TeX}}}

\copyrightyear{2020}
\acmYear{2020}
\setcopyright{rightsretained}
\acmConference[KDD '20]{Proceedings of the 26th ACM SIGKDD Conference on Knowledge Discovery and Data Mining}{August 23--27, 2020}{Virtual Event, CA, USA}
\acmBooktitle{Proceedings of the 26th ACM SIGKDD Conference on Knowledge Discovery and Data Mining (KDD '20), August 23--27, 2020, Virtual Event, CA, USA}
\acmDOI{10.1145/3394486.3403215}
\acmISBN{978-1-4503-7998-4/20/08}

\usepackage{booktabs}
\usepackage{subcaption}
\usepackage{soul}
\usepackage[ruled,linesnumbered]{algorithm2e}
\usepackage{amsmath}

\usepackage{thmtools}
\usepackage{thm-restate}

\usepackage{tikz}
\usetikzlibrary{decorations.pathmorphing}
\usetikzlibrary{fit}
\usetikzlibrary{backgrounds}
\usepackage{adjustbox}
\usepackage{titlesec}

\titleformat{\subsubsection}
{\normalfont\normalsize\bfseries}{\thesubsubsection}{0.5 em}{}

\declaretheoremstyle[
spaceabove=6pt,
spacebelow=6pt,
headfont=\normalfont\bfseries,
notefont=\mdseries\bfseries,
notebraces={(}{)},
bodyfont=\normalfont\itshape,
postheadspace=1em,
headpunct={:}]{mystyle}

\usepackage{environ}
\newenvironment{longver}{}{}
\newenvironment{shortver}{}{}

\NewEnviron{killlong}{}
{}
{}

\begin{document}
\fancyhead{}
\title{Statistically Significant Pattern Mining with Ordinal Utility}

\author{Thien Q. Tran}
\affiliation{%
  \institution{University of Tsukuba, Riken AIP}
}
\email{thientquang@mdl.cs.tsukuba.ac.jp}

\author{Kazuto Fukuchi}
\affiliation{%
  \institution{University of Tsukuba, Riken AIP}
}
\email{fukuchi@cs.tsukuba.ac.jp}

\author{Youhei Akimoto}
\affiliation{%
  \institution{University of Tsukuba, Riken AIP}
}
\email{akimoto@cs.tsukuba.ac.jp}

\author{Jun Sakuma}
\affiliation{%
  \institution{University of Tsukuba, Riken AIP}
}
\email{jun@cs.tsukuba.ac.jp}

\renewcommand{\shortauthors}{Thien Q. Tran, et al.}

\begin{abstract}
  \input{abstract}
\end{abstract}





\maketitle
\newcommand{\p}[1]{\text{Pr}\big[#1\big]}
\newcommand{\cp}[2]{\text{Pr}\big[#1\text{ }|\text{ } #2\big]}
\newcommand{\ev}[2]{\text{E}_{#1}\big[#2\big]}
\newcommand{\e}[1]{\text{E}\big[#1\big]}
\newcommand{\cm}{, \text{ }}
\newcommand{\eref}[1]{(\ref{#1})}
\newcommand{\nn}{\nonumber}
\newcommand{\mbf}{\mathbf}
\newcommand{\fa}{\text{ for all }}
\newcommand{\cond}{\text{ } | \text{ }}

\newcommand{\mysection}{\section}
\newcommand{\mysubsection}{\subsection}
\newcommand{\mysubsubsection}{\subsubsection}
\newcommand{\mysubsubsubsection}[1]{\textbf{#1} \\}

\input{ch-intro/chapter-intro}
\input{ch-pastwork/chapter-pastwork}
\input{ch-background/chapter-background}
\input{ch-setting/chapter-setting}
\input{ch-method/chapter-method}
\input{ch-analysis/chapter-analysis}
\input{ch-experiment/chapter-experiment}

\input{ch-conclusion/chapter-conclusion}
\input{ch-acknowledge/chapter-acknowledge}

\renewcommand*{\bibfont}{\small}
\bibliographystyle{ACM-Reference-Format}
\bibliography{reference}

\clearpage
\appendix
\mysection{Appendix: Proof of Theorem~\ref{theorem:fwer}}\label{app:proof_fwer}
\input{ch-appendicies/proof_fwer}
\mysection{Appendix: Other proofs}\label{app:proofs}
\input{ch-appendicies/proof}
\mysection{Appendix: Reproducibility}\label{data}
\input{ch-appendicies/implementation}

\end{document}

%% file: abstract.tex
Statistically significant patterns mining (SSPM) is an essential and challenging data mining task in the field of knowledge discovery in databases (KDD), in which each pattern is evaluated via a hypothesis test. Our study aims to introduce a preference relation into patterns and to discover the most preferred patterns under the constraint of statistical significance, which has never been considered in existing SSPM problems. We propose an iterative multiple testing procedure that can alternately reject a hypothesis and safely ignore the hypotheses that are less useful than the rejected hypothesis. One advantage of filtering out patterns with low utility is that it avoids consumption of the significance budget by rejection of useless (that is, uninteresting) patterns. This allows the significance budget to be focused on useful patterns, leading to more useful discoveries.

We show that the proposed method can control the familywise error rate (FWER) under certain assumptions, that can be satisfied by a realistic problem class in SSPM.\@We also show that the proposed method always discovers a set of patterns that is at least equally or more useful than those discovered using the standard Tarone-Bonferroni method SSPM.\@Finally, we conducted several experiments with both synthetic and real-world data to evaluate the performance of our method. As a result, in the experiments with real-world datasets, the proposed method discovered a larger number of more useful patterns than the existing method for all five conducted tasks.

%% file: ch-intro/chapter-intro.tex
\begin{figure}[t]
    \centering
    \includegraphics[width=0.65\linewidth]{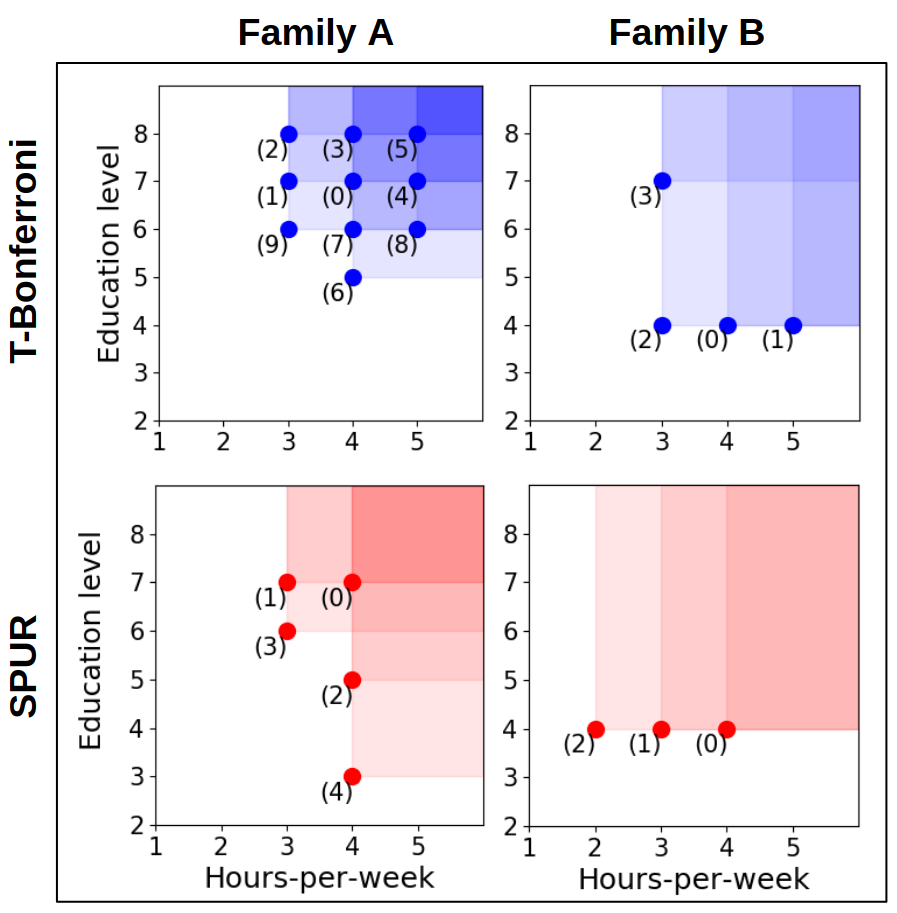}
    \vspace{-0.5em}
    \caption{Difference in discovered patterns and sorted indexes of p-values for two family A = (Male, Self-emp, Prof-specialty) and B = (Male, Private, Exec-managerial) by the existing method T-Bonferroni and our method SPUR.}\label{fig:rejections}
    \vspace{-1em}
  \end{figure}

\mysection{Introduction\label{ch:intro}}
Statistically significant pattern mining (SSPM) is the task of finding patterns that statistically occur more often in the data for one class than for another. Different from other pattern mining frameworks, in SSPM, the discovery of each pattern is evaluated via a statistical hypothesis test: a process to obtain a p-value which quantifies the probability that the association observed in the data is due to chance. The goal of SSPM is to maximize the number of true discoveries, i.e., minimize type-II error while controlling the number of false discoveries, i.e., controlling type-I error. With its statistical guarantee, SSPM is important and widely applied in fields such as genetics, healthcare, and market analysis. For example, in healthcare, discoveries of treatment combinations that are significantly efficacious could improve healthcare service quality. Similarly, medical scientists are interested in finding significant patterns of gene alleles associated with the onset of disease.


In this study, we focus on the utility of patterns, an aspect that is not considered in existing SSPM works. Pattern utility is essential in many applications spanning medicine, finance, and e-commerce, among others~\cite{gan2018survey}. For example, assume that we have the Adult dataset in which each transaction contains several demographic attribute values as explanatory variables and a binary target class (if income $>50$K or not). In general, SSPM aims to find patterns of demographic attributes (e.g., ``university graduate'' and ``works 60 hours/week'') that are significantly associated with ``income $>50$K''. Herein, we further introduce a preference relation between patterns and aim to find out the most useful patterns under the constraint of statistical significance. Continuing with the same example, assume that we are more interested in patterns that realize ``higher income with less education'' than ``higher income with higher education''; we can define patterns with ``less education'' as being more preferred than patterns with ``higher education''.

This paper serves to propose SPUR (Significant Pattern mining with Utility Relationship), a method that aims to discover statistically significant patterns with the highest utility. There are three main challenges that must be solved to achieve this goal.

First, our goal is different from the standard goal of SSPM.\@ Generally, SSPM aims to discover as many patterns as possible, while controlling type-I error.  Our goal is to discover the patterns that have the highest utility while still controlling type-I error. Therefore, existing SSPM methods, such as Tarone-Bonferroni (T-Bonferroni), are not necessarily practical for our setting. Since these methods do not consider preferences for different patterns, they may waste the significance budget for controlling the type-I error of patterns which are not useful. Here, we propose that after discovering a pattern, we can ignore all the patterns that are less useful than the discovered pattern. We expect that by ignoring less useful patterns, we can discover more useful patterns without violating the type-I error constraint.

Fig.~\ref{fig:rejections} demonstrates the statistically significant patterns in the Adult dataset found by the existing T-Bonferroni method (top) and our SPUR method (bottom). Specifically, we show the discovered patterns of ``education level'' and ``work hour'' while fixing other demographic features. We defined  ``less education'' and ``less hours-per-week'' as being more preferred than ``higher education'' and ``more hours-per-week''. As we see from the figure, in two different demographic conditions, our method can find more preferred (i.e., ``less education and less hours-per-week'' but ``higher income'') patterns than the existing method with fewer rejections.

Second, since each discovery in SSPM is evaluated via a statistical test, the proposed approach is only valid if we can develop an iterative multiple-testing procedure that can safely ignore less useful hypotheses after rejecting a hypothesis. By ``safely'', we mean that it does not violate the type-I error constraint. To the best of our knowledge, no existing multiple-testing procedure can fulfill this requirement. We discovered that this is achievable for a specific SSPM problem class. We propose a method that can control the FWER under a particular assumption about the independence of p-values, in that the p-values of false null hypotheses and true null hypotheses are mutually independent. In the SSPM context, this assumption can be satisfied when Fisher's exact test is used, and the setup of patterns guarantees that there is no overlap of samples between distinct patterns (see Section~\ref{assumption} for more details). We note that such a setting fits well for most of our intended applications.

Finally, because we often have to consider numerous patterns in the SSPM task, the number of hypotheses in the multiple-testing problem is also large. Generally, when the number of tests is large, simple multiple-testing adjustments, such as the Bonferroni correction, become too conservative for discovering significant patterns. Recently, many studies focus on the Tarone's trick~\cite{tarone1990modified}, which takes account of ``untestable'' hypotheses, resulting in more discoveries while still controlling the FWER.\@ Thus, to be practical in the SSPM setting, our method must also be able to leverage the Tarone's trick.

 Our contributions in this work are as follows:
\begin{itemize}
    \item We introduce the problem of discovering statistically significant patterns with the highest utility, given the ordinal utility between patterns. Then, we propose an iterative multiple-testing method that can reject more useful hypotheses and carefully ignore less useful hypotheses at each iteration.
    \item We prove that the proposed method can control the FWER under a particular assumption (independence of p-values between true and false hypotheses), which can be satisfied in a realistic SSPM problem class.\@
    \item We prove that the proposed method can always achieve a result that has a higher utility compared to T-Bonferroni, a standard SSPM method.\@
    \item We conduct several experiments using both synthetic and real-world datasets. For the real-world experiment, our method achieves a discovery result with higher utility for all five conducted tasks.
\end{itemize}

The remainder of the paper is organized as follows. We introduce some related works and fundamental concepts in Section~\ref{ch:pastwork} and~\ref{ch:background}. Then, we formally define our problem setting in Section~\ref{ch:setting}. Details concerning the proposed method are provided in Section~\ref{ch:method}. We prove the theoretical guarantee for FWER and improvement of utility in Section~\ref{ch:analysis}. Next, we demonstrate that the proposed method outperforms other methods in discovering useful patterns in Section~\ref{ch:experiment}, using both synthetic and real-world data. Finally, brief conclusions are put forward in Section~\ref{ch:conclusion}.

%% file: ch-pastwork/chapter-pastwork.tex
\mysection{Related Work\label{ch:pastwork}}
The most fundamental challenge in SSPM is the explosion in the number of patterns due to the number of factors. Some methods have been proposed to overcome this challenge in terms of improving discovery  power and navigating complexity. Early works in this respect were proceeded in~\cite{hamalainen2012kingfisher, webb2007discovering}. Then, Limitless Arity Multiple-testing Procedure (LAMP) - a method which can efficiently discover significant patterns in the higher-order were proposed in~\cite{terada2013statistical}. LAMP is designed with a combination of the Tarone's trick\cite{tarone1990modified} and the association rule mining algorithm Apriori\cite{agrawal1994fast}. Other studies have attempted to improve or extend LAMP to other settings\cite{terada2016significant, minato2014fast}.

There is also a body of literature focused on other types of tests or other aspects of the problem. For example,~\cite{llinares2015fast} studied the Westfall-Young permutation test for the purpose of dealing with the dependence between patterns.~\cite{llinares2017genome} worked on Cochran-Mantel-Haenszel and~\cite{pellegrina2019spumante} focused on Barnard's exact test. Moreover,~\cite{komiyama2017statistical} focused in the statistical emerging pattern mining problem,~\cite{webb2016multiple} considered the hypotheses stream problem, while~\cite{sugiyama2019finding} focused on finding significant interactions between continuous features. However, to the best of our knowledge, no SSPM studies have hitherto focused on the utility of patterns.

In pattern mining, many studies have focused on the utility of patterns, namely utility-oriented pattern mining (UPM). UPM is an essential task with numerous applications in finance, medicine, and e-commerce, among others\cite{gan2018survey}. Many UPM approaches have been proposed, expanding to various subfields, including high-utility item sets\cite{yao2006mining}, high-utility rules\cite{lee2013utility}, and maximal high utility\cite{shie2012efficient}. However, no studies on UPM have offered a statistical guarantee for discovered patterns.

%% file: ch-background/chapter-background.tex
\mysection{Preliminaries\label{ch:background}}
\mysubsection{Statistical association testing}\label{SSPM}
Statistical association testing is a procedure for testing whether two random variables are statistically dependent, or in other words, associated. In the context of pattern mining, this is a procedure to test whether the presence of a pattern is related with the occurrence of a specific event, represented by a class label. Suppose that we have a dataset $D = \{t_1, t_2, \ldots, t_{n_D}\}$ that contains $n_D$ transactions defined in the universe of $m$ items $I = \{1, \ldots, m\}$. Here, a transaction $t_i$ can be described by a vector $\mbf{x}_i$ of length $m$ and a binary-class label $y_i$. The vector $\mbf{x}_i$ indicates whether the corresponding items present in the transaction; i.e., $x_{ij}=1$ if item $j$ appear in transaction $t_i$. We call a set of items $s \subseteq I$ a pattern and define the indication variable $X_{i,s}=\prod_{j \in s} x_{ij}$ for pattern $s$. Here, $X_{i,s}=1$ if pattern $s$ appears in transaction $t_i$ and $X_{i,s} = 0$ otherwise.

The association of pattern $s$ and target class $y=1$ can be investigated by conducting an independence test with the null hypothesis $H_0\text{: }X_s \bot y$ via a $2\times2$ contingency table as in Table~\ref{tab:contingency}. Here, $n_D$ is the total number of transactions, and $n_1$ is the number of transactions with label $y=1$. Moreover, $n_s$ is the support for pattern $s$; i.e., the number of transactions that contain pattern $s$ and $a_s$ is the support for pattern $s$ among class $y=1$. A widely used independence test in SSPM is Fisher's exact test, which is a conditioned test in which $n_D, n_1, n_s$ are assumed to be fixed\cite{fisher1950statistical}. Under the null hypothesis of no association between $X_{s}$ and $y$, the count cell $a_s$ follows a hypergeometric distribution $P(. \mid n_D, n_1, n_s)$. Thus, the probability of observing the current table can be calculated as follows:

\begin{align*}
  P(a_s \mid n_D, n_1, n_s) = \frac{\binom{n_D}{a_s} \binom{n_D - n_1}{n_T - x_T}}
  {\binom{n_D}{n_T}}.
\end{align*}

Consequently, the p-value, i.e., the probability of observing a contingency table that is equally or more extreme as the observed table under the null hypotheses, can be obtained as
\begin{align*}
  p_{\text{val}}^{(s)} = \sum_{k=a_s}^{\min\{n_D, n_s\}} P(k | n_D, n_1, n_s).
\end{align*}

If the p-value $p_{\text{val}}^{(s)} \leq \alpha$ holds for some significance level $\alpha$, we can reject the null hypothesis of no association, and conclude that pattern $s$ is associated with  outcome $y$ under the significance level $\alpha$. Through the significance testing procedure, the probability of falsely rejecting a true null hypothesis, i.e., the probability of falsely discovering a false pattern, is controlled under the desired significance level $\alpha$.

\begin{table}[t]
  \centering
  \caption{$2\times2$ contingency table for pattern S}\label{tab:contingency}
  \vspace{-0.5em}
  \resizebox{0.6\linewidth}{!}{%
    \begin{tabular}{cccc}
      \toprule
                     & $X_{i,s} = 1$ & $X_{i, s} = 0$            & \textit{Total} \\
      \cmidrule(lr){1-1}\cmidrule(lr){2-3}\cmidrule(lr){4-4}
      $y_i = 1$      & $a_s$         & $n_1 - a_s$               & $n_1$          \\
      $y_i = 0$      & $n_s - a_s$   & $n_D - n_s - (n_1 - a_s)$ & $n_D-n_1$      \\
      \cmidrule(lr){1-1}\cmidrule(lr){2-4}\cmidrule(lr){4-4}
      \textit{Total} & $n_s$         & $n_D - n_s$               & $n_D$          \\
      \toprule
    \end{tabular}}
\end{table}

\mysubsection{Multiple testing}\label{multiple-testing}
In the previous section, we described an individual test for a pattern. However, in association mining, we have to consider many patterns, which means that multiple hypotheses have to be tested at the same time. If each test is conducted independently with a significance level of $\alpha$, the probability that at least one false discovery occurs would be much larger than $\alpha$. Typically, in such a case, it is necessary to control the overall error of all tested hypotheses. Two criteria are widely used in this respect: the familywise error rate (FWER) and false discovery rate (FDR). In what follows, we focus on controlling the FWER, which is defined as the probability of making at least one false rejection. Assume that we want to test a hypotheses set $H = \{h_1, \ldots, h_{|H|}\}$. Letting $T \subseteq H$ be the subset of true null hypotheses and $R \subseteq H$ be the set of hypotheses that were rejected by the multiple-testing procedure, FWER is defined as follows:
\begin{restatable}[Familywise error rate (FWER)]{definition}{deffwer}\label{def:fwer}
  \begin{align*}
    \text{FWER} = \p{R \cap T \neq \emptyset}
  \end{align*}
\end{restatable}

The most straightforward method to control the FWER is the Bonferroni correction, which uses a corrected rejection threshold $\delta = \alpha / m$ for each test. However, when the number of hypotheses is large, the Bonferroni correction can become too conservative. Recently, in the context of SSPM, many methods starting with LAMP\cite{terada2013statistical} have leveraged the Tarone's trick\cite{tarone1990modified} to exclude untestable hypotheses that will never be significant. To be specific, for tests in which the test statistics are discrete, we can evaluate the minimal attainable p-value for that test. When $\psi(h)$ is the minimal attainable p-value, hypothesis $h$ will never be rejected by a threshold $\sigma$ if $\psi(h) > \sigma$. In the case of Fisher's exact test, the minimal attainable p-value for a pattern $s$ is obtained as
\begin{align*}
  \psi(s) = p(k \mid n_D, n_1, n_s) \text{ where } k = \min \{n_1, n_s\}.
\end{align*}

Using $\psi(s)$, regardless of the count cell $a_s$, if $\psi(s) > \delta$, the hypothesis related to pattern $s$ will never be significant for significance level $\delta$. In SSPM, $n_1$ can be considered as fixed because the data are commonly given beforehand. Moreover, since $n_s = \sum_i {X_{i,s}}$, pattern $s$ can be ignored if there are too few transactions that contain pattern $s$, i.e., if the support for that pattern is too small. Since the number of hypotheses to be considered decreases with Tarone's trick, a larger rejection threshold can be used and results in more discoveries. Formally, the T-Bonferroni method leverages the Tarone's trick by setting the rejection threshold $\sigma_\mathrm{Tarone}$ as follows:
\begin{align*}
  \sigma_\mathrm{Tarone} = \max \{\sigma \mid \sigma |\kappa(\sigma)| \leq \alpha\}.
\end{align*}
where $\kappa(\sigma) = \{h \mid \psi(h) \leq \sigma\}$ is the testable hypothesis set regarding the rejection threshold $\sigma$.

\mysubsection{Assumption and the target problem class}\label{assumption}
Next, we introduce the key assumption that is necessary for our method followed by the target problem class of the proposed method.

\begin{restatable}[p-values independence]{assumption}{asindependence}\label{as:independence}
  For any hypotheses pair $h_t \in T$ and $h_f \in F$, their p-values $p_{h_t}$ and $p_{h_f}$ are independent.
\end{restatable}

In other words, we assume that the p-values obtained from the true patterns and the false patterns are mutually independent\footnote{Stated differently, we assume that the test statistics obtained from the true patterns and the false patterns are mutually independent and the p-values are determined using these test statistics.}. We next claim that Assumption~\ref{as:independence} holds for certain types of SSPM problems by the following proposition.

\begin{restatable}{proposition}{prousablecase}\label{pro:usable-case}
  In the setting of Section~\ref{SSPM} and~\ref{multiple-testing}, suppose p-values of Fisher's exact test are considered. Given $D$, if for any two distinct patterns $s, s' \in S$, $\{t \in D \mid s \in t\} \cap \{t \in D \mid s' \in t\} = \emptyset$ holds, Assumption 3.1 holds.
\end{restatable}
\begin{proof}
  \vspace{-0.3em}
  Appendix~\ref{app:proofs}.
  \vspace{-0.2em}
\end{proof}

Proposition~\ref{pro:usable-case} claims that Assumption~\ref{as:independence} can be satisfied by two requirements. First, it requires that Fisher's exact test is used, which is a popular independent test in SSPM contexts.\@ Second, it requires that the pattern set is designed to separate the dataset into non-overlapping subsets. This second requirement can be satisfied in several scenarios, for example
\begin{itemize}
\item Categorical dataset: consider a dataset with several categorical attributes $(x_1, x_2, \ldots, x_m)$ and a target label $y$. Assumption~\ref{as:independence} holds by any pattern set $S \subseteq I_1 \times \cdots \times I_m$, where $I_d$ is the set of possible values for variable $x_d$.
\item Transaction dataset with fixed transaction size: consider a dataset that for any transaction $t_i \in D$, the number of items in $t_i$ is $k$ for a fixed $k$, i.e., $\sum_{j \in I} x_{ij} = k$. Assumption~\ref{as:independence} holds by any pattern set $S \subseteq I^k$.
\end{itemize}

As a counter-example, Assumption~\ref{as:independence} does not hold for a set $S$ which contains both pattern $\{a\}$ and pattern $\{a, b\}$, since a transaction that contains $\{a, b\}$ will also contain $\{a\}$. It is noted that the categorical dataset setup fits well for most of our intended applications and is adopted for our real-world experiment in Section~\ref{ch:experiment}.

%% file: ch-setting/chapter-setting.tex
\mysection{Problem setting\label{ch:setting}}
To formally define our problem setting, we first introduce the concept of ordinal utility and clarify our goal of finding significant patterns with the highest utility. Next, we define a criterion for evaluating the goodness of the discovered result and discuss the limitation of existing approaches in our setting.

\mysubsection{Ordinal utility and the dominating subset}
We focus on the ordinal utility between patterns, i.e., which of two options is better. This differs from cardinal utility, which would consider how good the two options are and thus how much better one option is compared to the other. We chose ordinal utility for our setting because, in SSPM, assigning a utility score for each pattern is not always practical for patterns with multiple items. By contrast, using (partial) ordinal utility, we can easily define a flexible utility relationship between patterns, possibly using their items. We use $s_1 \succ s_2$ and $s_1 \approx s_2$ to denote that ``pattern $s_1$ is useful to pattern $s_2$'' and ``pattern $s_1$ is equally useful as pattern $s_2$'', respectively. Moreover, since each pattern $s$ is evaluated via a hypothesis test $h_s$ in SSPM, we similarly use $h_{s_1} \succ h_{s_2}$ and $h_{s_1} \approx h_{s_2}$ to denote  the ordinal utility between hypotheses. A preference order of utility for a pattern set $S$ is defined in the following.

\begin{restatable}[Preference order of utility]{definition}{defpreferenceorder}
  A preference order $\prec$ on a pattern set $S = \{s_1, \ldots, s_{|S|}\}$ is a transitive binary relation in which $(S, \prec)$ is a partially ordered set.
\end{restatable}

Such a preference order can be predefined based on the background knowledge or preferences of the user. For example, suppose that we are considering a dataset on medication usage in which each transaction includes a combination of drugs used by a patient and a binary class \textit{cured} or \textit{not-cured}. Moreover, we assume that we can define the total cost $cost(s)$ and the adverse effect level $adv(s)$ for each drug combination pattern $s$ as ordinal levels. A user who prefers patterns with a lower cost and less adverse effects can define the ordinal utility between patterns as
\begin{align*}
  cost(s_1) = cost(s_2) \land adv(s_1) = adv(s_2) & \iff  h_{s_1} \approx h_{s_2} \\
  cost(s_1) < cost(s_2) \land adv(s_1) \leq adv(s_2) & \implies  h_{s_1} \succ h_{s_2} \\
  cost(s_1) \leq cost(s_2) \land adv(s_1) < adv(s_2) & \implies  h_{s_1} \succ h_{s_2}.
\end{align*}

\begin{figure}[h]
  \centering
  \vspace{-1em}
  \includegraphics[width=0.5\linewidth]{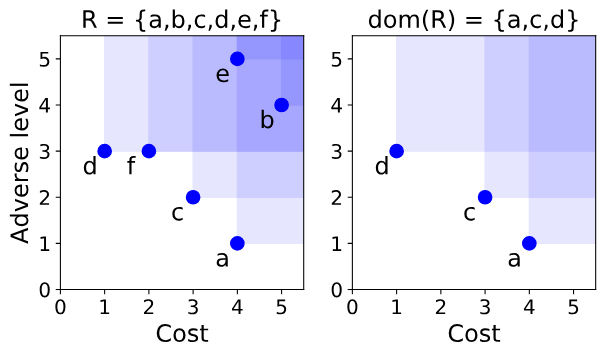}
  \vspace{-1em}
  \caption{Preference order and the dominating set}\label{fig:illustrate-drug}
\end{figure}

We show an example of a pattern set $R = \{a, b, c, d, e, f\}$ in Figure~\ref{fig:illustrate-drug} where each point $a, \ldots, f$ represents the utility of that pattern, i.e., the cost and the adverse effect levels. The colored rectangle next to each point represents the area that is less useful than that pattern. From the foregoing definition, we have that $h_b \prec h_d$, $h_e \prec h_c$, etc. We notice that there are pairs of patterns such as $h_c$ and $h_d$ in Figure~\ref{fig:illustrate-drug} that are not comparable. In other words, we do not require the preference order $\prec$ to be complete. This design choice is especially useful for the pattern mining setting, where each pattern contains multiple items, and multiple aspects of the pattern can be considered.

In SSPM, the utility of a pattern can be defined using the items included in that pattern. For example, in the medication example, let $cost(t)$ and $adv(t)$ be the ordinal level of the cost and the adverse effect level for a drug $t$. We can define a utility function $U: S \to N^d$ as
\begin{align*}
U(s)
= \begin{pmatrix} cost(s) \\ adv(s) \end{pmatrix}
= \begin{pmatrix} \sum_{t \in s} cost(t) \\ \max_{t \in s} adv(s) \end{pmatrix}.
\end{align*}

Next, we recall that our goal is to discover significant patterns with the highest utility. Hence, the goodness of a discovered pattern set is determined only by the patterns with the highest utility in that set. We call such a subset the (utility) dominating subset.

\begin{restatable}[(Utility) dominating subset]{definition}{defdominatingset}\label{def:dominating_set}
  For a pattern set $K = \{s_1, \dots, s_k\}$, we call $dom(K) \subseteq K$ the (utility) dominating subset of $K$ if $dom(K) = \{s \in K \cond \nexists_{s' \in K}, s \prec s'\}$
\end{restatable}

We also illustrate this concept by showing the dominating set $dom(R)$ for $R = \{a,b,c,d,e,f\}$ in Figure~\ref{fig:illustrate-drug} (right). We see that since $b \prec h$ (and also $c, d$), we have $b \notin dom(R)$. Similarly, we also have that $e \notin dom(R)$ and $d \notin dom(R)$ and we can obtain the dominating subset $dom(R) = \{a, c, d\}$.

Our setting also requires that the discovered patterns must be true patterns. That is, letting the set of true patterns (the set of false null hypotheses) be $F$, the best discovery result that can be achieved is $dom(F)$. However, in the SSPM setting, the number of discoveries is often limited because we have to control the number of false discoveries. Hence, it is often not possible to achieve $dom(F)$, especially when the number of true patterns is large. Thus, the practical goal in our setting is to discover a pattern set whose utility is as close as $dom(F)$ as possible.

\mysubsection{Criterion for utility evaluation}
Here, we introduce a metric to compare the goodness of two pattern sets, or in other words, two rejected hypothesis sets. This is then used to compare the utilities between discovered set $dom(R)$ and optimal set $dom(F)$. We also use this criterion to compare the results of different methods, which is required in the real-world setting, with unknown $F$.

\begin{restatable}[Utility measure]{definition}{defdistance}\label{def:distance}
  Given two pattern sets $K$ and $K'$, the utility measure from $K'$ to $K$ is denoted by $D_u(K \| K')$.
  \begin{align*}
    D_u(K \| K') = |\{h \in dom(K) \cond \nexists_{h' \in dom(K')}, h \preceq h'\}|.
  \end{align*}
\end{restatable}

Intuitively, $D_u(K \| K')$ is the number of more useful patterns in $dom(K)$ that are not included in $dom(K')$. We note that $D_u(K \| K')$ can be asymmetric, that is, can be different from $D_u(K' \| K)$. In Figure~\ref{fig:illustrate-distance}, we illustrate this utility measure for two sets $K = \{a, b, c, d\}$ and $K' = \{b, c, e, f\}$. In this case, $dom(K) = \{a, c, d\}$ and $dom(K') = \{c, f\}$. Thus, $D_u(K \| K') = 2$ and $D_u(K' \| K) = 0$.

\begin{restatable}[Utility comparison]{definition}{defcomparation}~\label{def:comparation}
  Given two pattern sets, $K$ and $K'$, we use $K \succ K'$ to denote that $K$ is more useful than $K'$, while $K \approx K'$ denotes that $K$ is equally as useful as $K'$. Here, the utility comparison is defined as follows:
  \begin{align*}
    D_u(K \| K') = 0 \land D_u(K' \| K) = 0 &\iff K \approx K' \\
    D_u(K \| K') > 0 \land D_u(K' \| K) = 0 &\iff K \succ K' \\
    D_u(K \| K') > 0 \land D_u(K' \| K) > 0 &\iff \text{not comparable}.
  \end{align*}
\end{restatable}

\begin{figure}[h]
  \centering
  \vspace{-1em}
  \includegraphics[width=1.0\linewidth]{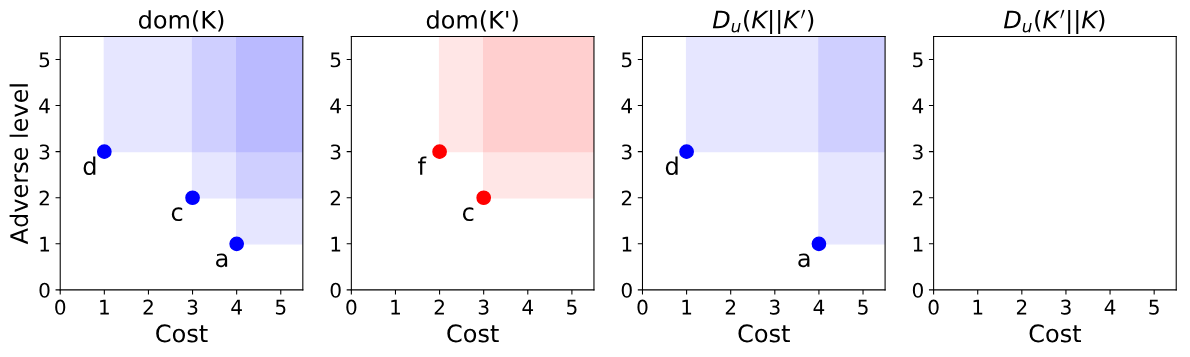}
  \vspace{-1em}
  \caption{Utility measure between two pattern sets $K=\{a, b, c, d\}$ and $K' = \{b, c, d, e\}$}\label{fig:illustrate-distance}
\end{figure}

It should be noted that there are cases such that two pattern sets are not comparable. Specifically, two pattern sets $K$ and $K'$ are not comparable if $D_u(K \| K') \neq \emptyset$ and $D_u(K' \| K) \neq \emptyset$. This is the case that $K$ includes some more useful patterns that are not included in $K'$ while at the same time $K'$ also includes some more useful patterns that are not included in $K$. Since our goal is to discover as many useful patterns as  possible, we consider both $K$ and $K'$ as profitable results, even when $D_u(K \| K') > D_u(K' \| K)$.

\subsection{Existing approaches and their limitations}\label{limitation}
The most simple approach to our problem is to use standard multiple-testing procedures such as Bonferroni or Holm\cite{holm1979simple} to obtain significant patterns, then filter out less useful patterns. The limitation of this approach is obvious: these procedures treat all hypotheses equally regardless of their utility. These methods are only developed to maximize the number of rejections $|R|$ under the constraint of FWER, i.e., to minimize type-II error while controlling type-I error. By contrast, our setting directly focuses on the dominating subset $dom(R)$, only. Consequently, from a utility perspective, many rejections of less useful patterns in $R$ is meaningless, resulting in a waste of the significance budget for controlling the type-I error of hypotheses which are not useful.

A better way forward may be a weighted approach, such as the weighted Bonferroni procedure\cite{rosenthal1983ensemble}. In this procedure, a weight is assigned to each hypothesis based on its importance. Hence, more important hypotheses have a higher chance of being rejected. However, if small weights are unfortunately assigned to false null hypotheses, many false null hypotheses will not be rejected, because the rejection thresholds are too strict. In the worst case, the weighted procedure might falsely accept all the hypotheses, i.e., miss all the true patterns. Considering that the weights must be assigned before conducting statistical tests to properly control FWER and no one knows which hypotheses are true null hypotheses, the weighted approach does not necessarily work well in practice.

%% file: ch-method/chapter-method.tex
\mysection{Proposed method\label{ch:method}}
\mysubsection{The proposed method: SPUR}~\label{sec:algorithm}
In our setting, after discovering a pattern set $R$, discovering any pattern $s$ that is less useful than a pattern in $R$, i.e., $\exists s' \in R, s' \succ s$, does not increase the result's utility. Thus, once a hypothesis $h_s$ is rejected, we can ignore all less useful hypotheses $h_{s'} \prec h_{s}$. Here, ``ignoring a hypothesis'' means ``accepting the hypothesis''. Because we have to guarantee the type-I error when rejecting a hypothesis, we can save the significance budget by rejecting more useful hypotheses and, at the same time, accepting less useful ones.

We propose an iterative multiple-hypotheses procedure that can conduct such an adaptive rejection process while controlling the FWER under Assumption~\ref{as:independence}. In particular, our method repeatedly (1) \textit{rejects the most significant hypothesis in the candidate hypothesis set} and (2) \textit{ignores all hypotheses in the candidate set that are less useful than the rejected hypothesis in the last step}. Next, we explain our algorithm along with the pseudocode provided in Algorithm~\ref{al:SPUR}. We also note that FWER control by our adaptive procedure is not obvious and we will discuss this in detail in Section~\ref{fwer}.

\begin{algorithm}[t]
  \KwData{Hypothesis set $H$, statistical level $\alpha$}
  \KwResult{Reject set $R$}

  $t \leftarrow 1, R_0 \leftarrow \emptyset$\;
  $H_1 \leftarrow H, \delta_1 \leftarrow \alpha, p_0^r \leftarrow 0$\;
  \While{$H_t \neq \emptyset$} {
    $\sigma_t \leftarrow \max \{ \sigma : (\sigma - p_{t-1}^r) |\kappa_t(\sigma)| \leq \delta_t \}$ \tcp*[r]{threshold}
    $p_t^r \leftarrow \min_{h \in H_t}p_h$\;
    $h_t^r \leftarrow \text{argmin}_{h \in H_t} p_h$\;
    \eIf{$p_t^r \leq \sigma_t$} {\label{al:rejection-rule}
      $R_{t+1} \leftarrow R_t \cap \{h_t^r\}$ \tcp*[r]{reject}
      $H_{t+1} \leftarrow  \{h \in H_t \cond h \succ h_t^r\}$ \tcp*[r]{ignore unuseful}
      $\tau_t \leftarrow \delta_t / (\sigma_t - p_{t-1}^r)$\;
      $\delta_{t+1} \leftarrow \delta_t - \tau_t (p_t^r - p_{t-1}^r) + p_t^r$ \tcp*[r]{modify budget}\label{al:modify-delta}
    }{ \Return{$R_{t-1}$}\; }
    $t \leftarrow t+1$
  }
  \Return{$R_{t-1}$}\;
  \caption{SPUR}\label{al:SPUR}
\end{algorithm}

\textbf{Initialization and notation:} First, we initialize the significance budget and the rejection set as $\delta_1=\alpha, R_0=\emptyset$. Here, $\delta_t$ is used to control the FWER of the procedure at iteration $t$. We define $H_t$ as the set of the remaining candidate hypotheses at iteration $t$ and let $H_1 = H$. Furthermore, we also define $\kappa_t(\sigma) = \{h \in H_t \cond \psi(h) \leq \sigma\}$ as the set of testable hypotheses at the $t$-th iteration for a rejection threshold $\sigma$.

\textbf{Rejection procedure:} At each iteration $t$, we decide whether to reject a hypothesis in the candidate set $H_t$ as follows.
\begin{enumerate}
\item We obtain the smallest p-value $p_t^r = \min_{h \in H_t} p_h$ and the corresponding hypothesis $h_t^r = \text{argmin}_{h \in H_t} p_h$. We also assume that $h_t^r$ can be decided by a predefined rule if there is more than one hypothesis realizing the smallest p-value.
\item We obtain $\sigma_t = \max\{\sigma: (\sigma-p_{t-1}^r) |\kappa_t (\sigma)| \leq \delta_t\}$ as the rejection threshold to decide whether to reject $h_t^r$.
\item If $p_t^r \leq \sigma_t$, we reject $h_t^r$ by adding it to the rejection set, i.e., $R_{t+1} = R_t \cap \{h_t^r\}$. We also ignore (accept) all less useful hypotheses in $H_t$ by setting the next candidate set as $H_{t+1} = \{h \in H_t \cond h \succ h_t^r\}$. Finally, we modify the significance budget as $\delta_{t+1} = \delta_t - \tau_t (p_t^r - p_{t-1}^r) + p_t^r$, where $\tau_t = \delta_t / (\sigma_t - p_{t-1}^r)$.
\item If $p_t^r > \sigma_t$, we stop and return the rejection set $R = R_{t-1}$.
\end{enumerate}

\mysubsection{An intuitive explanation\label{sec:intuitive}}
In SPUR, the rejection threshold $\sigma_t$ and significance budget $\delta_t$ are managed in a way that is not so intuitive. We will prove that such updates are necessary to control the FWER properly in Section~\ref{ch:analysis}. Before providing the formal proof, we give a more intuitive account of our rejection rule by considering a simple case: the case without the Tarone's trick. In other words, we assume that $\forall h \in H: \psi(h) = 0$. We have the following proposition.

\begin{restatable}[SPUR in limited case]{proposition}{prospurlimit}\label{pro:spur_limit}
  Assume for any $h \in H$, $\psi(h) = 0$, then the rejection rule of SPUR at step~\ref{al:rejection-rule} can be rewritten as
  \begin{align*}
    p_t^r \leq \frac{\alpha - \sum_{i=1}^{t-1} p_i^r(|H_{i}| - |H_{i+1}| - 1)}{|H_t|} &\implies \text{ reject } h_t^r \\
    otherwise &\implies \text{stop}.
  \end{align*}
\end{restatable}
\begin{proof}
  Appendix~\ref{app:proofs}.
\end{proof}

In our simple case, after ignoring less useful hypotheses, we decrease the significance budget by an amount of $p_i^r(|H_{i}| - |H_{i+1}| - 1)$. We note that $|H_{i}| - |H_{i+1}| - 1$ is exactly the number of ignored hypotheses due to the rejection of $h_i^r$. This term came from the modification of $\delta_t$ after each rejection, i.e., step~\ref{al:modify-delta} in Algorithm~\ref{al:SPUR}. Moreover, in Section~\ref{ch:experiment}, we demonstrate that this modification is necessary to properly control FWER by showing that the rejection threshold $\sigma_t = \frac{\alpha}{|H_t|}$ fails to control the FWER.\@ We also remark that when no hypotheses can be ignored, i.e., $\forall i \cm p_i^r(|H_{i}| - |H_{i+1} - 1|) = 0$, our method reduces to the step-down Holm method\cite{holm1979simple}. The full SPUR procedure is obtained by leveraging the Tarone's trick to additionally consider the testability of the candidate hypotheses set at each iteration.

%% file: ch-analysis/chapter-analysis.tex
\mysection{Theoretical analysis\label{ch:analysis}}
\mysubsection{FWER guarantee}
First, we show that the proposed method can control the FWER properly under Assumption~\ref{as:independence} by Theorem~\ref{theorem:fwer}. A more detailed proof is relegated to Appendix~\ref{app:proof_fwer}.\footnote{The completed proof is available at https://github.com/dizzyvn/SPUR/}

\begin{restatable}[FWER controlling]{theorem}{theoremfwer}\label{theorem:fwer}
  Under Assumption~\ref{as:independence}, the proposed  SPUR procedure conducted with significance level $\alpha$ controls the familywise error rate at $\alpha$.
\end{restatable}

Here, we discuss why Assumption~\ref{as:independence} is needed. We first remark that we reject the hypothesis $h_t^r$ at iteration $t$ if $p_t^r \leq \sigma_t$. In SPUR, because $\sigma_t$ relies on the p-values of the rejected hypotheses in the previous steps, it is burdensome to directly evaluate the probability of rejecting a true hypothesis, due to the complicated dependence of p-values between hypotheses. This characteristic makes the evaluation of FWER challenging for our procedure. We overcame this challenge using the following observations.

First, we remark that our testing framework is an iterative procedure. Thus, we can evaluate FWER by considering the probability of SPUR firstly rejecting a true hypothesis at an iteration. Actually, we show that FWER is equivalent to the sum of such a probability for all possible iterations. By viewing FWER this way, when evaluating the probability that SPUR firstly rejects a true hypothesis at iteration $T_E$, we can regard that all rejected hypotheses in previous steps $t < T_E$ are false hypotheses. Moreover, any variables at any iteration $t < T_E$, including $\delta_t, \sigma_t, H_t, h_t^r, p_t^r$, are determined only by the p-values of false hypotheses $\{p_h\} _{h \in F}$. Furthermore, $H_{T_E}, \delta_{T_E},$ and $\sigma_{T_E}$ are also determined only by $\{p_h\} _{h \in F}$ since these variables are obtained using only the variables of iteration $T_E - 1$. Hence, the only variable that relies on both the p-values of true and false hypotheses is $h_{T_E}^r$. Then, using Assumption~\ref{as:independence}, we can apply the Tarone's trick to control the probability of rejecting a true hypothesis in $H_{T_E}$ with due regard to the threshold $\sigma_{T_E}$ and show that FWER is upper-bounded by significant level $\alpha$.

Finally, it should be noted that our method can also be used as a general multiple testing framework for any problem that satisfies Assumption~\ref{as:independence}, such as multi-center studies or subset analysis in which the statistics between hypotheses are independent\cite{benjamini2000adaptive}.

\mysubsection{Utility guarantee}\label{fwer}
Next, we give a guarantee on the utility of the proposed method compared to the T-Bonferroni method in Theorem~\ref{theorem:lamp}.
\begin{restatable}[Comparison with T-Bonferroni]{theorem}{theoremlamp}\label{theorem:lamp}
  For arbitrary preference order $\prec$, let $R$ be rejected by SPUR and $R_\mathrm{Tarone}$ be the rejection set by T-Bonferroni; then $R \succeq R_\mathrm{Tarone}$.
\end{restatable}
\begin{proof}
  \vspace{-0.3em}
  Appendix~\ref{app:proofs}.
  \vspace{-0.2em}
\end{proof}

To show this theorem, we first claim that in SPUR, we always obtain a rejection threshold $\sigma_{t}$ that is not smaller than the rejection threshold of the last step $\sigma_{t-1}$.
\begin{restatable}[Monotonically increasing $\sigma_t$]{lemma}{lemsigmaincrease}\label{lem:sigmaincrease}
  Let $\sigma_0 = 0$, for any iteration $t \leq T_\mathrm{SPUR}$ of the SPUR process: $\sigma_{t-1} \leq \sigma_{t}$.
\end{restatable}
\begin{proof}
  \vspace{-0.3em}
  Appendix~\ref{app:proofs}.
  \vspace{-0.2em}
\end{proof}

Recall that the proposed method always rejects the most significant hypothesis $h_t^r$ then ignores other less useful hypotheses in $H_t$ regarding $h_{t}^{r}$. Thus, the \textit{necessary condition} for a hypothesis $h^* \in H$ to be ignored is that there exists another hypothesis $h \in H$ that (1) is more significant than $h^*$, i.e., $p_h < p_h^*$, and (2) is at least as useful as $h^*$, i.e., $h \succeq h^*$. In our proof, we show that this condition does not hold for any hypothesis $h^*$ included in the dominating set $dom(R_\mathrm{Tarone})$ discovered by T-Bonferroni. Moreover, using Lemma~\ref{lem:sigmaincrease} and the fact that the $\sigma_1 = \sigma_{Tarone}$, we can show that SPUR always rejects all the hypotheses in $dom(R_\mathrm{Tarone})$. Thus, $R \supseteq dom(R_\mathrm{Tarone})$ and Theorem~\ref{theorem:lamp} can then be shown using the utility comparison between hypotheses sets.

This guarantee of utility improvement is the most critical advantage of SPUR compared to the weighted approaches. As discussed in Section~\ref{limitation}, the discovery result by weighted approaches can be heavily affected if the false hypotheses are assigned small weights. In the worst case, the weighted approach would be inferior to methods that do not consider utility, e.g., T-Bonferroni. In contrast, although SPUR is also expected to achieve better performance when the utility of the false hypotheses is high, it still guarantees a utility that is not less useful than T-Bonferroni even when the utility of false hypotheses is low.

%% file: ch-experiment/chapter-experiment.tex
\mysection{Experiment\label{ch:experiment}}
In this section, we evaluate SPUR using a synthetic experiment and three real-world datasets.

\mysubsection{Synthetic experiment}\label{synthetic-experiment-1}
We conduct the first synthetic experiment with two goals. First, we verify if SPUR can adequately control FWER. Particularly, we demonstrate that adjusting the significance budget by $\delta_{t+1} = \delta_t - \tau_t(p_t^r - p_{t-1}^r) + p_t^r$ is necessary to control the FWER by showing that obtaining the significance budget as $\delta_{t+1} = \delta_t$ would violate the FWER.\@ Second, we verify if SPUR can correctly reject hypotheses with higher utility compared to other methods. We confirm our discussion in Section~\ref{fwer} about the limitation of the weighted approaches.

\mysubsubsection{Experiment setting:}
We consider a set of $100$ hypotheses $H=\{h_1, h_2, \dots, h_{100}\}$, where $20$ of them are false null hypotheses. Let $F$ be the set of false hypotheses and let the preference be $h_i \succ h_j$ for $i < j$, i.e., the hypothesis with a smaller index be more useful. We consider the following three settings where each is named by the usefulness of the false hypotheses set.

\begin{tabular}{ll}
  (1) $F$ has High utility  &: $F = \{h_1, h_2, \dots, h_{20}\}$ \\
  (2) $F$ has Medium utility&: $F = \{h_1, h_{6}, \dots, h_{96}\}$ \\
  (3) $F$ has Low utility   &: $F = \{h_{81}, h_{82}, \dots h_{100}\}$.
\end{tabular}

Since we do not consider the Tarone's trick in this experiment, we adopt the z-test and set up the null hypothesis $H_0: \mu = 0$ for all hypotheses. Moreover, we generate a dataset $D_i=\{x_{ij}\} _{j=1}^{20}$ with $20$ samples for each hypothesis $h_i$, where $x_{ij} \sim \mathcal{N}(\mu = 0, \sigma=0.75)$ if $h_i \in T$ and $x_{ij} \sim \mathcal{N}(\mu = 0.5, \sigma=0.75)$ if $h_i \in F$. We remark that the p-values obtained in this setting satisfy Assumption~\ref{as:independence} although this setting does not employ Fisher's exact test.

\mysubsubsection{Comparative methods and criterion}
We compare our method with the following three methods.
\begin{itemize}
  \item Bonferroni: the T-Bonferroni method with $\sigma = \alpha / |H|$
  \item w-Bonferroni: the weighted Bonferroni method where hypothesis $h_i$'s weight $w_i$ is assigned to be the number of hypotheses that are less useful than $h_i$, i.e., $w_i = 100 - i + 1$.
  \item invalid-SPUR:\@ the version of SPUR where $\delta_{t+1} = \delta_t$
\end{itemize}
For each setting, we generated $10{,}000$ datasets, applied each generated dataset to the four methods, and then evaluated the FWER and the utility of the rejected set $R$ (only for the runs with no type-I error). To evaluate utility, we first define $rank(h)$ as the utility ranking of a false hypothesis $h$ within the false hypothesis set $F$. For example, in the \textit{medium} utility setting, we have $rank(h_1)=1, rank(h_6)=2$. Here, the smaller the $rank(.)$, the more useful that rejected set is. The utility of the rejected set $R$ is the ranking of the most useful hypothesis in $R$.

\mysubsubsection{Results and discussion}
\begin{table}[t]
  \centering
  \caption{FWER and the average number of rejections.}\label{tab:fwer}
  \vspace{-1em}
  \resizebox{1.0\linewidth}{!}{%
  \begin{tabular}{lcccccc}
    \toprule
                 & \multicolumn{3}{c}{FWER} & \multicolumn{3}{c}{Average number of rejects}   \\
    \cmidrule(lr){2-4} \cmidrule(lr){5-7}
                 & High                     & Medium         & Low   & High  & Medium & Low   \\ \midrule
    SPUR         & 0.006                    & 0.042          & 0.048 & 3.157 & 2.265  & 1.934 \\
    Bonferroni   & 0.039                    & 0.041          & 0.041 & 3.493 & 3.488  & 3.480 \\
    w-Bonferroni & 0.032                    & 0.038          & 0.049 & 4.504 & 3.199  & 1.803 \\
    invalid-SPUR & 0.006                    & \textbf{0.056} & 0.048 & 3.268 & 2.378  & 1.947 \\
    \bottomrule
   \end{tabular}}
   \vspace{-0.5em}
\end{table}

The FWER and the average number of rejections according to the four methods are given in Table~\ref{tab:fwer}. From the table, we observe that Bonferroni, w-Bonferroni, and SPUR succeed in controlling the FWER in all three settings. On the other hand, invalid-SPUR failed to control the FWER for the \textit{medium} utility setting. This result demonstrates that management of the significance budget $\delta_t$ at line~\ref{al:modify-delta} of SPUR is necessary for controlling the FWER.\@

\begin{figure}[h]
  \centering
  \vspace{-1em}
  \includegraphics[width=0.95\linewidth]{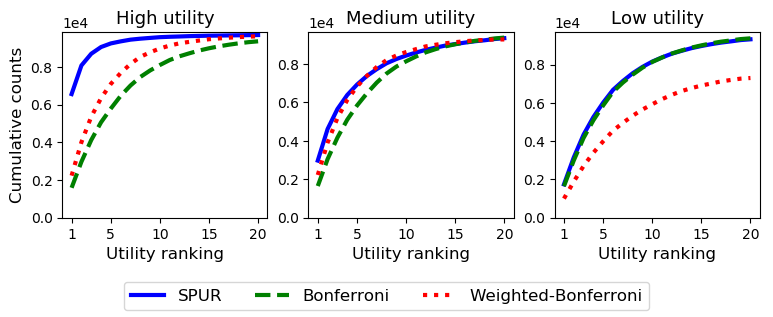}
  \vspace{-1em}
  \caption{Cumulative histogram of utility rank.}\label{fig:usefulness}
  \vspace{-1em}
\end{figure}

Figure~\ref{fig:usefulness} is a cumulative histogram of the utility rank $rank(R)$ of the rejection set $R$ by the three methods: Bonferroni, w-Bonferroni, and SPUR after $10{,}000$ runs. For each possible value $k \in \{1, \dots, 20\}$ of $rank(R)$, we count how many times each method obtained a rejection set $R$ so that $rank(R) \leq k$. For example, considering $k=2$ in the \textit{medium} utility setting, we count how many times each method rejected hypothesis $h_1$ or $h_6$ correctly. In Figure~\ref{fig:usefulness}, the x- and y-axes represent $k$ and the count of correct rejection at value $k$, respectively.

First, in all the settings, SPUR always achieved a rejection set that is not less useful than the Bonferroni rejection set, as is guaranteed in Theorem~\ref{theorem:lamp}. Moreover, we observed that SPUR could achieve a rejection set with a higher utility if the false hypotheses have high utility. This is because when the utility of false hypotheses is high, once a false hypothesis gets rejected, SPUR can ignore many non-useful hypotheses, which are both less useful and less significant than the rejected false hypothesis. As a result, SPUR can achieve a rejected set with higher utility.

On the other side, w-Bonferroni, which assigns larger weights for hypotheses with higher utility, works well in the \textit{high} and \textit{medium} utility settings. However, in the \textit{low} utility setting, this method performs poorly and even worse than the Bonferroni. This can be explained in terms of the dependence between the assigned weight and the p-value in the weighted approach. In particular, if the weights of the false hypotheses are small, their rejection thresholds could become too strict, and these hypotheses could not be rejected. In contrast, SPUR guarantees that the utility of the rejected set is at least useful as the result of Bonferroni even in the \textit{low} utility setting.

\mysubsection{Real-world datasets}~\label{real-data}
We conduct real-world data experiments to confirm if SPUR can discover more useful patterns in a real-world situation where the Tarone's trick must be considered. We adopted the three datasets Adult\cite{data-adult}, Crash Report\cite{data-crash}, and Crime\cite{data-crime} for five mining tasks.

\mysubsubsection{Data preprocessing and task setup}
All three datasets consist of some explanatory variables and a target class. To elaborate, the Crash Report dataset\cite{data-crash} is a dataset of traffic accidents, including information about speed limits, weather, light-conditions, etc., and a class $y=1$ indicates if someone was injured. The Crime dataset\cite{data-crime} consists of criminal records with information about the place, street, and time along with a class for the crime category: against a person, property, or society. For this dataset, we set up three tasks using each of these categories as the target class. Finally, the Adult dataset\cite{data-adult} contains several demographic attributes and a class $y=1$ of income $>50$K for different individuals.

To set up the candidate pattern set and the null hypotheses in each dataset, we focused on some variables and defined the mining task based on these variables. We translate the values of these variables into items by categorizing them with some predefined rules. For example, the \textit{hours-per-week} attribute in the Adult dataset is categorized into $[<20], [20-29], \dots, [\geq 60]$. A valid pattern is a combination of all explanatory variables, with one item for each variable. This setup guarantees that there is no overlap of samples between patterns, and we adopted the Fisher's exact test for all five tasks (one-sided for the Crash and Adult datasets, two-sided for the Crime dataset). Our goal is to discover the combinations of levels that are significantly associated with the target class. Details about the adopted variables are shown in Table~\ref{tab:variable-description} of Appendix~\ref{data}.

To define the ordinal utility between patterns, we first divide the explanatory variables into two groups: utility variables and family variables. Specifically, we state $K \succ K'$ if their family variables are identical, and the utility variables of $K$ is more useful than of $K'$. For example, for the Adult dataset, we let \textit{(sex, work-class, and occupation)} be the family variables and \textit{hours-per-week, education} be the utility variables. Here, based on prior knowledge that ``a higher income correlates with higher education level and more work hours'', we oppositely prefer patterns with \textit{lower education level} and \textit{less work hours}. That is, we aim to discover patterns that are unexpectedly associated with the class $>50$K. Furthermore, the family variables are also effective in finding useful patterns in various situations of \textit{(sex, work-class, and occupation)}. In addition, we prefer \textit{sounded safe} for the Crash task and \textit{occurs in midday} patterns for the Crime dataset.

\mysubsubsection{Results and discussion}
\begin{table}[t]
  \centering
  \caption{Number of patterns discovered by T-Bonferroni $|R_T|$ and SPUR $|R_S|$ along with the utility measure $D_u = D_u(R_S \| R_T)$, for different $\alpha$}\label{tab:num-reject}
  \vspace{-1em}
  \resizebox{1.0\linewidth}{!}{%
  \begin{tabular}{lrrrrrrrrr}
    \toprule
    & \multicolumn{3}{c}{$\alpha=0.01$} & \multicolumn{3}{c}{$\alpha=0.05$} & \multicolumn{3}{c}{$\alpha=0.1$}\\
    \cmidrule(lr){2-4} \cmidrule(lr){5-7} \cmidrule(lr){8-10}
                 & $|R_T|$ & $|R_S|$ & $D_u$       & $|R_T|$ & $|R_S|$ & $D_u$      & $|R_T|$ & $|R_S|$ & $D_u$       \\ \midrule
        Crash    & 83      & 75      & \textbf{12} & 120     & 94      & \textbf{6} & 136     & 106     & \textbf{13} \\
        Adult    & 45      & 26      & \textbf{0}  & 57      & 32      & \textbf{2} & 63      & 36      & \textbf{3}  \\
        Property & 275     & 142     & \textbf{10} & 315     & 156     & \textbf{8} & 336     & 164     & \textbf{10} \\
        Person   & 133     & 73      & \textbf{2}  & 156     & 83      & \textbf{2} & 163     & 87      & \textbf{1}  \\
        Society  & 320     & 142     & \textbf{3}  & 348     & 156     & \textbf{7} & 361     & 164     & \textbf{9}  \\
    \bottomrule
    \end{tabular}}
    \vspace{-1em}
\end{table}

In Table~\ref{tab:num-reject}, we show the number of discoveries $|R_\mathrm{S}|$ by SPUR, $|R_\mathrm{T}|$ by T-Bonferroni, and the utility measure $D_u(R_\mathrm{S} \| R_\mathrm{T})$, for different significance levels $\alpha$. We do not show $D_u(R_\mathrm{T} \| R_\mathrm{S})$, because $D_u(R_\mathrm{T} \| R_\mathrm{S}) = 0 $ for all settings, i.e., because SPUR always discovered a pattern set that is not less useful than T-Bonferroni, as is guaranteed by Theorem~\ref{theorem:lamp}. Especially, $D_u$ is large for the Crash, Property, and Society tasks, indicating that SPUR discovered a large number of high-utility patterns that cannot be discovered by T-Bonferroni.

On the other hand, the utility measure $D_u(R_\mathrm{S} \| R_\mathrm{T})$ is comparatively small for the Adult and Person tasks. We explain this by considering the number of discovered patterns $|R_T|$ by T-Bonferroni. First, in these tasks, $|R_T|$ is small compared to other tasks. Moreover, even when we increase the significance level $\alpha$, only a few additional patterns are discovered. In such a case, since the number of discoverable patterns is small in the first place, even if SPUR can achieve a larger rejection threshold by ignoring less useful hypotheses, the number of newly discovered patterns would not increase. In other words, for cases such as the Crash, Property, and Society tasks, SPUR is especially functional because many useful true patterns were not discovered by T-Bonferroni due to the FWER constraint, while there are many patterns in $R_T$ which are not useful.

Next, we focus on the discovery process of two methods by using the Adult task as an example. In Figure~\ref{fig:rejections}, we show the discoveries by two methods with $\alpha = 0.05$ for two families (Male, Self-emp, Prof-specialty) and (Male, Private, Exec-managerial). We also show the sorted indexes of p-values for discovered patterns by each method. As we can see from Figure~\ref{fig:rejections}, by considering both the significance and the utility of discovered patterns, SPUR can efficiently expand its dominance of patterns with fewer members of rejection. By contrast, with no consideration of utility, T-Bonferroni wastes its significance budget in rejecting less useful hypotheses. As a result, SPUR can discover more useful patterns that T-Bonferroni fails to discover. This advantage is not limited to discovering patterns in the same family, but is also helpful in discovering patterns of other families. In fact, we obtained a number of pattern families that only SPUR could discover.

In Figure~\ref{fig:rejections}, as discovered by SPUR, a male (private) executive/manager who works $40$--$50$ hours a week is likely to have an income $>50$K even if he has just graduated from high school, while the required education level found by T-Bonferroni is to graduate from a college. Moreover, a male who graduated from a professional school and is working in a (self-employed) professional specialty is likely to have an income of $>50$K even if he only works for $30$--$40$ hours a week, while the requirement discovered by T-Bonferroni is $40$--$50$ hours a week. Other than those patterns, in the Crash dataset, SPUR discovered many patterns that seemed safe but are still associated with an accident and injury, which cannot be discovered by T-Bonferroni. In this way, SPUR can discover patterns that are significant, and at the same time, preferred by the user.

%% file: ch-conclusion/chapter-conclusion.tex
\mysection{Conclusion\label{ch:conclusion}}

\input{ch-conclusion/conclusion} 
\input{ch-conclusion/futurework} 

%% file: ch-conclusion/conclusion.tex
In this study, we focused on the utility of patterns in the SSPM context. We introduced the problem of discovering statistically significant patterns with the highest utility, giving the ordinal utility between patterns. We proposed an iterative multiple-testing framework that alternately rejects a hypothesis and safely ignores hypotheses which are less useful than the rejected one. This enables the discovery of more useful patterns. We theoretically showed the FWER guarantee under a particular assumption and the utility guarantee of the proposed method. Finally, we conducted several experiments with both synthetic and real-world datasets to demonstrate that the proposed method is capable of discovering more useful patterns under the constraint of type-I error.

%% file: ch-conclusion/futurework.tex

%% file: ch-acknowledge/chapter-acknowledge.tex
\begin{acks}
This work was partly supported by KAKENHI (Grants-in-Aid for Scientific Research) Grant Numbers JP19H04164 and JP18H04099. We are also grateful to Professor Takeuchi for providing useful discussions and advices.
\end{acks}

%% file: ch-appendicies/proof_fwer.tex
To show this theorem, we first define the stochastic process obtained by the proposed method SPUR. 
\begin{restatable}[SPUR process]{definition}{defspurprocess} \label{def:spur-process}
  \begin{longver}
    Let $p_0^r = 0, R_1 = \emptyset, \delta_1 = \alpha, H_1 = H$, a stochastic process $\{X_t\}_{t=1}^{T_\mathrm{SPUR}}$ with $X_t = (\delta_t, R_t, p_t^r, h_t^r, \sigma_t, \tau_t)$ is said to be a SPUR process with stopping time $T_\mathrm{SPUR}$ if:
    \begin{align}
      \text{for } t > 1: \quad    
      &H_{t}       = \{h \in H | \forall_{h^r \in R_{t-1}}, h > h^r\} \\
      &\delta_{t}  = \delta_{t-1} - \tau_{t-1}(p_{t-1}^r - p_{t-1}^r) + p_{t-1}^r \\
      \text{for } t \geq 1: \quad 
      & p_t^r   = \min_{h \in H_t} p_h \cm h_t^r   = \text{argmin}_{h \in H_t} p_h \\
      & \sigma_t= \max\{\sigma : (\sigma- p_{t-1}^r)|\kappa_t(\sigma|) \leq \delta_t\} \\
      & \tau_t  = \delta_t / (\sigma_t - p_{t-1}^r) \\
      & R_{t}   = 
      \begin{cases}
        R_{t-1} \cap \{h_{t-1}^r\} \quad & (p_t^r \leq \sigma_t) \\
        R_{t-1} \quad & (otherwise)
      \end{cases}                                
    \end{align}
    and $T_\mathrm{SPUR} = \min\{t \in [|H|]: p_t^r > \sigma_t\}$
  \end{longver}

  \begin{shortver}
    A SPUR (stochastic) process is a stochastic process $\{X_t\}_{t=1}^{T_\mathrm{SPUR}}$ stopped at stopping time $T_\mathrm{SPUR}$ where $X_t = (\delta_t, R_t, p_t^r, h_t^r, \sigma_t, \tau_t)$ are obtained from the algorithm SPUR and $T_\mathrm{SPUR} = \min\{t \in [|H|]: p_t^r > \sigma_t\}$.
  \end{shortver}  
\end{restatable}

We remark that since we can obtain $H_t = \{h \in H | \forall_{h^r \in R_{t-1}}, h \succ h^r\}$, we do not need to include $H_t$ in $X_t$. Using the stochastic process $\{X_t\}_{t=1}^{T_\mathrm{SPUR}}$, we next rewrite the event of rejecting at least one true hypotheses, i.e., the event of occuring a Type-I error.

\begin{restatable}{lemma}{lem:spur-bound} \label{lem:spur-bound}
  Consider a SPUR process $\{X_t\}_{t=1}^{T_\mathrm{SPUR}}$, let $E_t = \{h_t^r \in T, p_t^r \leq \sigma_t\}$  and $T_E = \min \{t \in [|H|]: t \leq T_\mathrm{SPUR} \land E_t\}$, the following holds:
  \begin{align*}
    \{R \cap T \neq \emptyset\} = \{T_E \leq |H|\}.
  \end{align*}  
\end{restatable}
\begin{longver}
  \begin{proof}
    \begin{align}
      &\{T_E = t_0\} \\
      = & \{t_0 \leq T_{SPUR}\} \cap (\cap_{t=1}^{t_0 - 1} \bar E_t) \cap E_{t_0} \\
      = & \{\forall_{t \leq t_0}, p_t^r \leq \sigma_t\} \cap (\cap_{t=1}^{t_0 - 1} \bar E_t) \cap E_{t_0} \\
      = & \{\forall_{t \leq t_0}, p_t^r \leq \sigma_t\} \cap \{\forall_{t < t_0}, h_t^r \in F \lor p_t^r > \sigma_t\} \\
      &\cap \{h_{t_0}^r \in T \land p_{t_0}^r \leq \sigma_{t_0}\}
    \end{align}

    Since, 
    \begin{align}
      &\{\forall_{t \leq t_0}, p_t^r \leq \sigma_t\} \\
      \implies \quad & \{\forall_{t < t_0}, h_t^r \in F \lor p_t^r > \sigma_t\} = \{\forall_{t < t_0}, h_t^r \in F\}    
    \end{align}

    Hence,
    \begin{align}      
      &\{T_E = t_0\} \\
      = & \{\forall_{t < t_0}, h_t^r \in F \land p_t^r \leq \sigma_t\} \cap \{h_{t_0}^r \in T \land p_{t_0}^r \leq \sigma_{t_0}\}
    \end{align}

    In the other hand, we consider the event of rejecting at least a true hypothesis
    \begin{align}
      &\{R \cap T \neq \emptyset\} \\
      &= \cup_{t_0=1}^{|H|}(\{T_\mathrm{SPUR} = t_0\} \cap \{\exists_{t \leq t_0}, h_t^r \in T\}) \\
      &= \cup_{t_0=1}^{|H|}(\{T_\mathrm{SPUR} = t_0\} \cap (\cup_{t=1}^{t_0} \{\forall_{i < t_0} h_i^r \in F, h_t^r \in T\}) \\
      &= \cup_{t_0=1}^{|H|}(\cup_{t=1}^{t_0} \{T_\mathrm{SPUR} = t_0\}) \cap \{\forall_{i < t} h_i^r \in F, h_t^r \in T\} \\
      &= \cup_{t_0=1}^{|H|} \{t_0 \leq T_\mathrm{SPUR}\} \cap \{\forall_{i < t_0} h_i^r \in F, h_{t_0}^r \in T\} \\
      &= \cup_{t_0=1}^{|H|} \{\forall_{i \leq t_0}, p_i^r \leq \sigma_t \} \cap \{\forall_{i < t_0}, h_i^r \in F, h_{t_0}^r \in T\} \\  
      &= \cup_{t_0=1}^{|H|} \{\forall_{t < t_0}, h_t^r \in F \land p_t^r \leq \sigma_t\} \cap \{h_{t_0}^r \in T \land p_{t_0}^r \leq \sigma_{t_0}\} \\
      &= \cup_{t_0=1}^{|H|} \{T_E = t_0\} = \{T_E \leq |H|\}
    \end{align}

    This concludes our proof.
  \end{proof}
\end{longver}
\input{ch-method/graph_spur}

Here, $E_t = \{h_t^r \in T, p_t^r \leq \sigma_t\}$ is the event of rejecting a true hypothesis at the $t$-th iteration and $T_E = \min \{t \in [|H|]: t \leq T_\mathrm{SPUR} \land E_t\}$ is the first iteration that a true hypotheses got rejected. Lemma \ref{lem:spur-bound} claims that the FWER is actually the probability of the first Type-I error occurs at some step $t_0 \leq |H|$. To show the FWER controlling, we have to consider the relationship between the p-values of the true hypotheses, the false hypotheses, and the rejection threshold at each iteration. However, such the dependence are complicated in the SPUR process $\{X_t\}_{t=1}^{T_\mathrm{SPUR}}$. Thus, we instead consider an alternative stochastic process that only depends on the false hypotheses and show the FWER controlling by analyzing this process. In addition, we illustrate the relation of entities in the SPUR process using a graph shown in Figure \ref{fig:graph-spur}.

\begin{restatable}[False hypotheses based process]{definition}{deffalseprocess} \label{def:false-process}
  \begin{longver}
    Let $p_0^* = 0, H_1^* = H, R_0^* = \emptyset, \delta_1^* = \alpha$, a stochastic process $\{Y_t\}_{t=1}^{T_\mathrm{FALSE}}$ with $Y_t = (\delta_t^*, R_t^*, H_t^*, p_t^{r*}, h_t^{r*}, \sigma_t^*, \tau_t^*)$ is said to be a false-hypotheses based SPUR process with stopping time $T_\mathrm{FALSE}$ if:
    \begin{align}
      \text{for } t > 1: \quad    
      &H_{t}^*     = \{h \in H_{t-1}^* \cond h > {h_{t-1}^{r*}}\} \\
      &\delta_{t}^*  = \delta_{t-1}^* - \tau_{t-1}(p_{t-1}^{r*} - p_{t-1}^{r*}) + p_{t-1}^{r*} \\
      \text{for } t \geq 1: \quad 
      & p_t^{r*}   = \min_{h \in H_t^* \cap F} p_h \cm h_t^{r*}   = \text{argmin}_{h \in H_t^* \cap F} p_h \\
      & \sigma_t^*= \max\{\sigma : (\sigma - p_{t-1}^*)|\kappa_t^*(\sigma)| \leq \delta_t^*\} \\
      & \tau_t^*  = \delta_t^* / (\sigma_t^* - p_{t-1}^*) \\
      & R_{t}^*   = 
      \begin{cases}
        R_{t-1}^* \cap \{h_{t-1}^*\} \quad & (p_t^{r*} \leq \sigma_t^*) \\
        R_{t-1}^* \quad & (otherwise)
      \end{cases}    
    \end{align}
      and $T_\mathrm{FALSE} = \min\{t \in [H]: p_t^* > \sigma_t^*\}$
  \end{longver}

  \begin{shortver}
    A false hypotheses based process is a stochastic process $\{Y_t\}_{t=1}^{T_\mathrm{FALSE}}$ stopped at stopping time $T_\mathrm{False}$ where $Y_t = (\delta_t^*, R_t^*, p_t^{r*}, h_t^{r*}, \sigma_t^*, \tau_t^*)$ are obtained from the algorithm SPUR with the following modification:
    \begin{align*}
      p_t^{r*} = \min_{h \in H_t^* \cap F} p_h \text{ and } h_t^{r*} = \text{argmin}_{h \in H_t^* \cap F} p_h.
    \end{align*}
  \end{shortver}
\end{restatable}

Next, we define a sequence of true hypotheses' p-values $\{p_t^{T*}\}_{t=1}^{T_\mathrm{False}}$ where each $p_t^{T*}$ is obtained using the value $Y_t$ of the false hypotheses based process.
\begin{restatable}[Alternative true hypotheses sequence]{definition}{alternative} \label{def:alternative-true}
  A set of r.v. $\{p_t^{T*}\}_{t=1}^{T_\mathrm{False}}$ is said to be an alternative true hypotheses sequence obtained from the false hypotheses based process $\{Y_t\}_{t=1}^{T_\mathrm{FALSE}}$ if for $t \leq T_\mathrm{FALSE}$:
  \begin{align*}    
    p_t^{T*} = f(Y_t, \{p_h\}_{h \in T}) =  \min_{h \in H_t^* \cap T} p_h.
  \end{align*}  
\end{restatable}

Since at each step, SPUR rejects the most significant hypothesis $h_t^r = min_{h \in H_t} p_h$, the event $\{h_t^r \in T\}$ (and $\{h_t^r \in F\}$) depends on the comparison of $min_{h \in H_t \cap F} p_h$ and $min_{h \in H_t \cap R} p_h$. We next consider this comparison via the false hypotheses based process and the alternative true hypotheses sequence, while claims its relation to Lemma \ref{lem:spur-bound}.

\begin{restatable}{lemma}{lem:alternative-bound} \label{lem:alternative-bound}    
  Consider $\{Y_t\}_{t=1}^{T_\mathrm{False}}$ and $\{p_t^{T*}\}_{t=1}^{T_\mathrm{False}}$ as defined in Definition \ref{def:false-process} and \ref{def:alternative-true}. 
  Let $E^*_t = \{p_t^{T*} \leq p_t^{r*}, p_t^{T*} \leq \sigma_t^*\}$ and $T_E^* = \min \{t \in [|H|]: t \leq T_\mathrm{False} \land E^*_t\}$ and $T_E$ as defined in Lemma \ref{lem:spur-bound},
  \begin{align*}
    &T_E \geq T_E^* \text{ almost surely}.
  \end{align*}  
\end{restatable}
\begin{longver}
  \begin{proof}
    \begin{align}
      &\{T_E^* = t_0\}\\
      = &\{t_0 \leq T_\mathrm{False}\} \cap (\cap_{t=1}^{t_0 - 1} \bar E_t^*) \cap E_{t_0}^* \\
      = &\{\forall_{t < t_0}, p_t^* \leq \sigma_t^*\} \cap \{\forall_{t < t_0}, p_t^{T*} > p_t^{r*} \lor p_t^{T*} > \sigma_t^*\}\\ 
      &\cap \{p_{t_0}^{T*} \leq p_{t_0}^{r*} \land p_{t_0}^{T*} \leq \sigma_{t_0^*}\} \\
      = &\{\forall_{t < t_0}, p_t^{T*} > p_t^{r*} \land p_t^{r*} \leq \sigma_t^*\} \cap \{p_{t_0}^{T*} \leq p_{t_0}^{r*} \land p_{t_0}^{T*} \leq \sigma_{t_0^*}\}
    \end{align}

    Next, we show that 
    \begin{align}
      \forall_{t < t_0}, p_t^{T*} > p_t^{r*} \implies \forall_{t < t_0}, X_t = Y_t \label{shared-prefix}
    \end{align}
    
    First, we have
    \begin{align}
      &p_t^{T*} > p_t^{r*}  \\
      \iff \quad &\min_{h \in H_t^* \cap T} > \min_{h \in H_t^* \cap F}\\
      \iff \quad &\text{argmin}_{h \in H_t^* \cap F} = \text{argmin}_{h \in H_t^*}
    \end{align}
    
    In the other hand,
    \begin{align}
      h_t^r \in F &\iff \min_{h \in H_t} = \min_{h \in H_t \cap F}\\
      h_t^r \in T &\iff \min_{h \in H_t} = \min_{h \in H_t \cap T}
    \end{align}

    Hence,
    \begin{align}
      &p_t^{T*} > p_t^{r*} \land H_t = H_t^* \\
      \implies \quad &\text{argmin}_{h \in H_t^* \cap F} = \text{argmin}_{h \in H_t} \\
      \implies \quad &h_t^{r*} = h_1^r \\
      \implies \quad &X_t = Y_t
    \end{align}

    Moreover, we have
    \begin{align}
      &p_t^{T*} > p_t^{r*} \land (H_t, \delta_t, \sigma_t, \tau_t) = (H_t^*, \delta_t^*, \sigma_t^*, \tau_t^*) \\
      \implies \quad & X_t = Y_t \\
      \implies \quad & (\delta_{t+1}, H_{t+1}, \sigma_{t+1}, \tau_{t+1}) = (\delta_{t+1}^*, H_{t+1}^*, \sigma_{t+1}^*, \tau_{t+1}^*)
    \end{align}

    The last line follows from the definition of two procedure. Moreover, since $(H_1, \delta_1, \sigma_1, \tau_1) = (H_1^*, \delta_1^*, \sigma_1^*, \tau_1^*)$, 
    \begin{align}
      p_1^{T*} > p_1^{r*} \implies X_t = Y_t
    \end{align}

    The claim \eref{shared-prefix} can now be shown using induction and we rewrite the event $\{T_E = t_0\}$ as
    \begin{align}  
      &\{T_E = t_0\} \\
      = &\{\forall_{t < t_0}, h_t^r \in F \land p_t^r \leq \sigma_t\} \cap \{h_{t_0}^r \in T \land p_{t_0}^r \leq \sigma_{t_0}\} \\
      = &\{\forall_{t < t_0}, h_t^{r*} < h_t^{T*} \land p_t^{r*} \leq \sigma_t^*\} \cap \{h_t^{T*} < h_{t_0}^{r*} \in T \land p_{t_0}^{r*} \leq \sigma_{t_0}^*\} \\
      \subseteq &\{\forall_{t < t_0}, h_t^{r*} < h_t^{T*} \land p_t^{r*} \leq \sigma_t^*\} \cap \{h_t^{T*} \leq h_{t_0}^{r*} \in T \land p_{t_0}^{r*} \leq \sigma_{t_0}^*\}
    \end{align}

    Thus, 
    \begin{align}
      \{T_E^* = t_0\} \implies \{T_E = t_0\}
    \end{align}

    This concludes our proof.
  \end{proof}
\end{longver}

\input{ch-method/graph_false}

We also give an illustration on the events $E_t^*$ and their relationship with other entities in Figure \ref{fig:graph-false}. We have that $\p{R \cap T \neq \emptyset} \leq \p{T_E \leq |H|} \leq \p{T_E^* \leq |H|}$ from Lemma \ref{lem:alternative-bound}. Moreover, since
\begin{align*}
  \p{T_E^* \leq |H|} = \ev{\{p_h\}_{h \in F}}{\p{T_E^* \leq |H| \cond \{p_h\}_{h \in F}}}.
\end{align*}

we next find the upper bound of $\p{T_E^* \leq |H| \cond \{p_h\}_{h \in F}}$.
\begin{restatable}{lemma}{lem:conditioned-bound} \label{lem:conditioned-bound}  
  Using the same definition of Lemma \ref{lem:alternative-bound} and let $k = T_\mathrm{False}$, under Assumption \ref{as:independence}, the following holds:  
  \begin{align*}
    &\p{T_E^* \leq |H| \cond \{p_h\}_{h \in F}} \\
    \leq &\sum_{t_0 < k} |\kappa_{t_0}^*(p_{t_0}^{r*}) \cap T| (p_{t_0}^{r*} - p_{t_0 - 1}^{r*}) + |\kappa_k^*(p_{k}^{r*}) \cap T| (\sigma_k^* - p_{k - 1}^{r*}).
  \end{align*}  
\end{restatable}
\begin{longver}
  \begin{proof}
    First, since the False-hypotheses based process $\{Y_t\}_{t=1}^{T_\mathrm{False}}$  determined only by only the values of $\{p_h\}_{h \in F}$, 
    \begin{align}
      &\p{T_E^* \leq |H| \cond \{p_h\}_{h \in F}} \\
      = &\p{T_E^* \leq |H| \cond \{p_h\}_{h \in F}, \{Y_t\}_{t=1}^{T_\mathrm{False}}} \\
      = &\p{T_E^* \leq |H| \cond p_F, Y_F, k}
    \end{align} 

    In the last line, we let $k = T_\mathrm{False}$, and compactly express the conditional term for the sake of space, where 
    $p_F$ and $Y_F$ represent $\{p_h\}_{h \in F}$ and $\{Y_t\}_{t=1}^{T_\mathrm{False}}$, respectively.
  
    From the definition of $T_E^*$, we have:
    \begin{align}
      &T_E^* = \min \{t \in [|H|]: t \leq T_\mathrm{False} \land E^*_t\} \\
      \implies \quad &T_E^* \leq T_\mathrm{False} \\
      \implies \quad &\p{T_E^* \leq |H| \cond p_F, Y_F, k} = \p{T_E^* \leq k \cond p_F, Y_F, k}
    \end{align}

    Beside, using the result of $T_E^*$ obtained before:
    \begin{align}
      &\{T_E^* \leq k\} \\
      = &\cup_{t_0 \leq k} \{\forall_{t < t_0}, p_t^{T*} > p_t^{r*} \land p_t^{r*} \leq \sigma_t^* \land p_{t_0}^{T*} \leq p_{t_0}^{r*} \land p_{t_0}^{T*} \leq \sigma_{t_0^*}\} \\
      \subseteq &\cup_{t_0 \leq k} \{p_{t_0 - 1}^{T*} > p_{t_0 - 1}^{r*} \land p_{t_0}^{T*} \leq p_{t_0}^{r*} \land p_{t_0}^{T*} \leq \sigma_{t_0^*}\} \\
      \subseteq &\cup_{t_0 \leq k} \{p_{t_0}^{T*} > p_{t_0 - 1}^{r*} \land p_{t_0}^{T*} \leq p_{t_0}^{r*} \land p_{t_0}^{T*} \leq \sigma_{t_0^*}\} 
    \end{align}
    
    The last line follows since $p_{t_0}^{T*} \geq p_{t_0-1}^{T*}$. We next consider the term $\{p_{t_0}^{T*} \leq p_{t_0}^{r*} \land p_{t_0}^{T*} \leq \sigma_{t_0}^*\}$ for different iteration $t_0$ regarding of the stopping time $k = T_\mathrm{FALSE}$. Using the definition of $T_\mathrm{FALSE}$:
    \begin{align}
      k = \min\{t \in [H]: p_t^* > \sigma_t^*\}  
    \end{align}

    Hence,
    \begin{align}
      t_0 < k &\Leftrightarrow p_{t_0}^{T*} \leq \sigma_{t_0}^* \\
              &\Rightarrow \{p_{t_0}^{T*} \leq p_{t_0}^{r*} \land p_{t_0}^{T*} \leq \sigma_{t_0}^*\} = \{p_{t_0}^{T*} \leq p_{t_0}^{r*}\} \\
      t_0 = k &\Leftrightarrow p_{t_0}^{T*} > \sigma_{t_0}^* \\
              &\Rightarrow \{p_{t_0}^{T*} \leq p_{t_0}^{r*} \land p_{t_0}^{T*} \leq \sigma_{t_0}^*\} = \{p_{t_0}^{T*} \leq \sigma_{t_0}^*\}
    \end{align}
    
    And then,
    \begin{align}
      &\{T_E^* \leq k\} \\
      \subseteq &\cup_{t_0 < k} \{p_{t_0}^{T*} > p_{t_0 - 1}^{r*} \land p_{t_0}^{T*} \leq p_{t_0}^{r*} \land p_{t_0}^{T*} \leq \sigma_{t_0}^*\} \\
      &\cup \{p_{k}^{T*} > p_{k - 1}^{r*} \land p_{k}^{T*} \leq p_{k}^{r*} \land p_{k}^{T*} \leq \sigma_{k^*}\} \\
      \subseteq &\cup_{t_0 < k} \{p_{t_0 - 1}^{r*}  < p_{t_0}^{T*} \leq p_{t_0}^{r*}\} \cup \{p_{k - 1}^{r*} < p_{k}^{T*} \leq \sigma_{k}^*\}
    \end{align}

    The probability of this event is 
    \begin{align}
      &\p{T_E^* \leq k \cond p_F, Y_F, k} \\
      \leq &\sum_{t_0 < k} \p{p_{t_0 - 1}^{r*}  < p_{t_0}^{T*} \leq p_{t_0}^{r*} \cond p_F, Y_F, k} \\ &+ \p{p_{k - 1}^{r*} < p_{k}^{T*} \leq \sigma_{k}^* \cond p_F, Y_F, k} \\
      \leq &\sum_{t_0 < k} \cup_{h \in \kappa_{t_0}^*(p_{t_0}^{r*}) \cap T} \p{p_{t_0 - 1}^{r*}  < p_h \leq p_{t_0}^{r*} \cond p_F, Y_F, k} \\ &+ \cup_{h \in \kappa_{k}^*(\sigma_k^*) \cap T} \p{p_{k - 1}^{r*} < p_h \leq \sigma_{k}^* \cond p_F, Y_F, k}
    \end{align}
    
    Moreover, since $\{Y_t\}_{t=1}^{k}$ only relies on $\{p_h\}_{h \in F}$, by the Assumption \ref{as:independence}:
    \begin{align}
      &\{p_h\}_{h \in F}, \{Y_t\}_{t=1}^{k} \bot \{p_h\}_{h \in F}
    \end{align}
    
    Then, since $p_h, h \in T$ follows the Uniform distribution, 
    \begin{align}
      &\p{T_E^* \leq k \cond p_F, Y_F, k} \\
      \leq &\sum_{t_0 < k} |\kappa_{t_0}^*(p_{t_0}^{r*}) \cap T| (p_{t_0}^{r*} - p_{t_0 - 1}^{r*}) + |\kappa_k^*(\sigma_k^*) \cap T| (\sigma_k^* - p_{k - 1}^{r*})
    \end{align}

    This concludes our proof.
  \end{proof}
\end{longver}

Actually, the proposed algorithm SPUR is designed to guarantee that the right side in the inequation of Lemma \ref{lem:conditioned-bound} always less than $\alpha$, as stated in Lemma \ref{lem:algoritm-bound}.
\begin{restatable}{lemma}{lem:algorithm-bound} \label{lem:algoritm-bound}
  Using the same definition of Lemma \ref{lem:alternative-bound} and let $k = T_\mathrm{False}$, the following holds:    
  \begin{align*}
    \sum_{t_0 < k} |\kappa_{t_0}^*(p_{t_0}^{r*}) \cap T| (p_{t_0}^{r*} - p_{t_0 - 1}^{r*}) + |\kappa_k^*(\sigma_t^*) \cap T| (\sigma_k^* - p_{k - 1}^{r*}) \leq \alpha.
  \end{align*}  
\end{restatable}
\begin{longver}
  \begin{proof}
    Let
    \begin{align}
      \Delta_k = \sum_{t_0 < k} |\kappa_{t_0}^*(p_{t_0}^{r*}) \cap T| (p_{t_0}^{r*} - p_{t_0 - 1}^{r*}) + |\kappa_k^*(\sigma_k^*) \cap T| (\sigma_k^* - p_{k - 1}^{r*})  
    \end{align}
    
    By our definition of the False-hypotheses based process, we have
    \begin{align}
      &\forall_{t_0 < k}, p_{t_0}^{r*} \leq \sigma_{t_0}^* \\
      \implies \quad &\forall_{t_0 < k}, \kappa_{t_0}^*(p_{t_0}^{r*}) \subseteq \kappa_{t_0}^*(\sigma_{t_0}^*)
    \end{align}
    
    Thus, 
    \begin{align}
      \Delta_k \leq \sum_{t_0 < k} |\kappa_{t_0}^*(\sigma_{t_0}^*) \cap T| (p_{t_0}^{r*} - p_{t_0 - 1}^{r*}) + |\kappa_k^*(\sigma_k^*) \cap T| (\sigma_k^* - p_{k - 1}^{r*})
    \end{align}  
    
    Moreover, 
    \begin{align}
      T_\mathrm{False} = k \implies \forall_{t_0 < k}, |\kappa_{t_0}^*(\sigma_{t_0}^*) \cap F| \geq k - t_0
    \end{align}
    
    And since,
    \begin{align}
      &\kappa_{t_0}^*(\sigma_{t_0}^*) = (\kappa_{t_0}^*(\sigma_{t_0}^*) \cap F) \cup (\kappa_{t_0}^*(\sigma_{t_0}^*) \cap T) \\
      \text{and }& F \cap T = \emptyset
    \end{align}

    The following holds
    \begin{align}
      \forall_{t_0 < k}, |\kappa_{t_0}^*(\sigma_{t_0}^*) \cap T| \leq |\kappa_{t_0}^*(\sigma_{t_0}^*)| - k + t
    \end{align}

    Hence,
    \begin{align}
      \Delta_k \leq &\sum_{t_0<k} (|\kappa_t^*(\sigma_{t_0}^*)| - k + t_0)(p_{t_0}^{r^*} - p_{{t_0}-1}^{r^*}) + |\kappa_k^*(\sigma_k^*)|(\sigma_k^* - p_{k-1}^{r^*}) \\
      \leq &\sum_{t_0 < k} \big[|\kappa_t^*(\sigma_{t_0}^*)|(p_{t_0}^{r^*} - p_{t_0-1}^{r^*}) + p_{t_0}^{r^*}\big] + |\kappa_k^*(\sigma_k^*)|(\sigma_k^* - p_{k-1}^{r^*})
    \end{align}
    
    Besides, since $|\kappa_{t_0}^*(\sigma)|(\sigma - p_{t_0 - 1}^{r^*}) \leq \delta_{t_0}^*$ for ${t_0} \leq k$,
    \begin{align}
      \forall_{t_0 < k}, & |\kappa_k^*(\sigma_{t_0}^*)| \leq \frac{\delta_{t_0}^*}{p_{t_0}^{r*} - p_{t_0 - 1}^{r^*}} = \tau_{t_0}^*  \\
      & |\kappa_k^*(\sigma_k^*)|(\sigma_k^* - p_{k-1}^{r^*}) \leq \delta_k^*
    \end{align}
    
    Thus,
    \begin{align}
      \Delta_k &\leq \sum_{t_0 < k} \big[\tau_t (p_t^{F^*} - p_{t-1}^{F^*}) + p_t^{F^*}\big] + \delta_k^* = \alpha
    \end{align}  
    
    This concludes our proof.
  \end{proof}
\end{longver}

Theorem \ref{theorem:fwer} is then shown by combining the above lemmas.

%% file: ch-method/graph_spur.tex
\begin{figure}[t]
  \centering
  \tikzstyle{state}=[circle,
                    thick,
                    minimum size=0.8cm,
                    draw=blue!80,
                    fill=blue!20]

  \tikzstyle{measurement}=[circle,
                          thick,
                          minimum size=0.8cm,
                          draw=orange!80,
                          fill=orange!25]

  \tikzstyle{input}=[rectangle,
                    thick,
                    minimum size=0.6cm,
                    draw=purple!80,
                    fill=purple!20,
                    rounded corners=1mm]

  \tikzstyle{background}=[rectangle,
  fill=gray!10,
  inner sep=0.2cm,
  rounded corners=5mm]
  
  \begin{adjustbox}{max size={.5\linewidth}}  
    \begin{tikzpicture}[>=latex,text height=1.3ex,text depth=0.25ex]
      
      \matrix[row sep=0.25cm, column sep=0.25cm] {
        & \node (p_F) [input]{$\{p_h\}_{h \in T}$}; &
        & \node (p_T) [input]{$\{p_h\}_{h \in F}$}; &
        \\
        \node (x_1) [state] {$X_1$}; &
        \node (x_dot_1) {$\cdots$}; &
        \node (x_t) [state] {$X_t$};  &
        \node (x_dot_2) {$\cdots$}; &
        \node (x_T) [state] {$X_{T_k}$}; &
        \\

        \node (e_1) [measurement] {$E_1$}; &
        \node (e_dot_1) {$\cdots$}; &    
        \node (e_t) [measurement] {$E_t$}; &
        \node (e_dot_1) {$\cdots$}; &    
        \node (e_T) [measurement] {$E_{T_k}$}; &
        \\
        };
        
        \path[->]
            (p_T) edge (x_1)
            (p_T) edge (x_t)
            (p_T) edge (x_T)        

            (p_F) edge (x_1)
            (p_F) edge (x_t)
            (p_F) edge (x_T)                

            (x_1) edge (e_1)
            (x_t) edge (e_t)
            (x_T) edge (e_T)

            (x_1) edge (x_dot_1)
            (x_dot_1) edge (x_t)
            (x_t) edge (x_dot_2)        
            (x_dot_2) edge (x_T)
        ;  
      \begin{pgfonlayer}{background}
        \node [background,
                    fit=(x_1) (x_T) (p_T) (p_F)
                    ] {};
      \end{pgfonlayer}          
    \end{tikzpicture}
  \end{adjustbox}

  \caption{Relation graph of $\{X_t\}_{t=1}^{T_\mathrm{SPUR}}$ (shaded) and event $E_t$ (we let $k=T_\mathrm{False}$ in the graph)}
  \label{fig:graph-spur}
  \vspace{-0.5em}  
\end{figure}
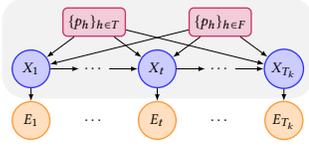

%% file: ch-method/graph_false.tex
\begin{figure}[t]
\centering
\tikzstyle{state}=[circle,
                  thick,
                  minimum size=0.8cm,
                  draw=blue!80,
                  fill=blue!20]

\tikzstyle{measurement}=[circle,
                        thick,
                        minimum size=0.8cm,
                        draw=orange!80,
                        fill=orange!25]

\tikzstyle{input}=[rectangle,
                  thick,
                  minimum size=0.6cm,
                  draw=purple!80,
                  fill=purple!20,
                  rounded corners=1mm]

\tikzstyle{background}=[rectangle,
                        fill=gray!10,
                        inner sep=0.2cm,
                        rounded corners=5mm]

\begin{adjustbox}{max size={.66\linewidth}}
  \begin{tikzpicture}[>=latex,text height=1.3ex,text depth=0.25ex]

  \matrix[row sep=0.25cm,column sep=0.25cm] {
    &&& \node (p_F) [input]{$\{p_h\}_{h \in F}$}; &
    \\
    \node (y_1) [state] {$Y_1$}; &&
    \node (y_dot_1) {$\cdots$}; &
    \node (y_t) [state] {$Y_t$};  &&
    \node (y_dot_2) {$\cdots$}; &
    \node (y_T) [state] {$Y_{T_k}$}; &
    \\

    \node (pt_1) [state] {$p_1^{T*}$}; &
    \node (e_1) [measurement] {$E_1^*$}; &

    \node (pt_dot_1) {$\cdots$}; &

    \node (pt_t) [state] {$p_1^{T*}$};  &
    \node (e_t) [measurement] {$E_t^*$}; &

    \node (pt_dot_2) {$\cdots$}; &

    \node (pt_T) [state] {$p_{k}^{T*}$}; &
    \node (e_T) [measurement] {$E_{T_k}^*$}; &
    \\

    &&& \node (p_T) [input]{$\{p_h\}_{h \in T}$}; &
    \\
    };

    \path[->]
        (p_F) edge (y_1)
        (p_F) edge (y_t)
        (p_F) edge (y_T)

        (p_T) edge (pt_1)
        (p_T) edge (pt_t)
        (p_T) edge (pt_T)

        (y_1) edge (pt_1)
        (y_t) edge (pt_t)
        (y_T) edge (pt_T)

        (y_1) edge (y_dot_1)
        (y_dot_1) edge (y_t)
        (y_t) edge (y_dot_2)
        (y_dot_2) edge (y_T)

        (y_1) edge (e_1)
        (pt_1) edge (e_1)

        (y_t) edge (e_t)
        (pt_t) edge (e_t)

        (y_T) edge (e_T)
        (pt_T) edge (e_T)
    ;

    \begin{pgfonlayer}{background}
      \node [background,
                  fit=(y_1) (y_T) (p_F)
                  ] {};
    \end{pgfonlayer}
\end{tikzpicture}
\end{adjustbox}

\caption{Relation graph of $\{Y_t\}_{t=1}^{T_\mathrm{False}}$ (shaded), $\{p_t^{T*}\}_{t=1}^{T_\mathrm{False}}$ and event $E_t^*$ (we let $k=T_\mathrm{False}$ in the graph)}
\label{fig:graph-false}
\vspace{-1em}
\end{figure}
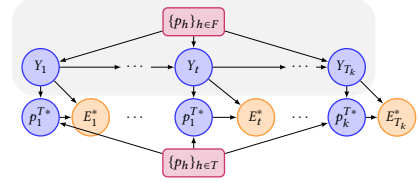

%% file: ch-appendicies/proof.tex
\mysubsection{Proof of Proposition~\ref{pro:spur_limit}}
\begin{proof}
  Since $\forall_{h \in H} \psi(h) = 0$, we have that
  \begin{align*}
 & \forall_t, \forall_{\sigma \geq 0}, \kappa_t(\sigma) = |\{h \in H_t| \psi(h) \leq \sigma\}| = |H_t|.
  \end{align*}

  We also have
  \begin{align*}
    \sigma_t              & = \max\{\sigma : (\sigma - p_{t-1}^r) |H_t| \leq \delta_t\} = \frac{\delta_t}{|H_t|} + p_{t-1}^r               \\
    \tau_t              & = \delta_t / (\sigma_t - p_{t-1}^r) = |H_t|                                                   \\
    \text{Thus, }\sigma_t & = \frac{\delta_t + p_{t-1}^r|H_t|}{|H_t|}                                                    \\
                     & = \frac{\alpha - \sum_{i=1}^{t-1} \left[(p_i^r - p_{i-1}^r)|H_i| - p_i^r\right] + |H_t|p_{t-1}^r}{|H_t|} \\
                     & = \frac{\alpha - \sum_{i=1}^{t-1} p_i^r(|H_i| - |H_{i+1}| - 1)}{|H_t|}.
  \end{align*}

  In the last line, we rewrite the formula by grouping terms for each $p_t^r$ and used the fact that $p_0^r = 0$. This concludes our proof.
\end{proof}

\mysubsection{Proof of Proposition~\ref{pro:usable-case}}
\begin{proof}
  Without the loss of generality, we consider two patterns $s$ and $s'$ where $h_s \in T$ and $h_s' \in F$. We adopt the notions in Section~\ref{ch:background} and assume that the observed dataset is $D = \{t_i, y_i\} _{i=1}^{n_D} \sim \mathcal{D}$. We represent this dataset regarding $s$ and $s'$ as $D = \{x_{i,s}, x_{i,s'}, y_i\}$ where $x_{i,s}=1$ if pattern $s$ is included in transaction $i$.

  Since Fisher's exact test assumes that the margin distribution $n_1, n_D$ and $n_s, n_{s'}$ are fixed, the only remaining r.v.\ in obtaining the p-value are $a_s$ and $a_{s'}$, where $a_s = \sum_{i \in [n_D], x_i^s = 1} y_i$ and $a_{s'} = \sum_{i \in [n_D], x_i^{s'} = 1} y_i$. Letting $I_s = \{i \in [n_D], x_i^s = 1\}$, $I_{s'} = \{i \in [n_D], x_i^{s'} = 1\}$, from the design of $s$ and $s'$, we have $I_s \cap I_{s'} = \emptyset$. Hence, $a_s$ and $a_{s'}$ are obtained using mutually distinguish transactions. We next consider the independence of $(x_i^s, y_i^s)$ and $(x_{i'}^{s'}, y_{i'}^{s'})$ for $i \neq i'$:
  \begin{align*}
    \p{x_{i'}^{s'}, y_{i'}^{s'} \mid x_i^s, y_i^s} & = \p{y_{i'}^{s'} \mid x_{i'}^{s'}, x_i^s, y_i^s} \p{x_{i'}^{s'} \mid x_i^s, y_i^s}                                   \\
                                                & = \frac{\p{y_{i'}^{s'}, x_{i'}^{s'}}}{\p{x_{i'}^{s'} \mid x_i^s, y_i^s}} \frac{\p{x_{i'}^{s'}, x_i^s}}{\p{x_i^s \mid y_i^s}} \\
                                                & = \frac{\p{y_{i'}^{s'}, x_{i'}^{s'}}}{\p{x_{i'}^{s'} \mid x_i^s}\p{x_i^s \mid y_i^s}} \frac{\p{x_{i'}^{s'}, x_i^s}}{\p{x_i^s \mid y_i^s}}.
  \end{align*}

  On the other hand, since $h_s$ is a true hypotheses, we have $x_i^s \bot y_i^s$.
  \begin{align*}
    \p{x_i^s, y_i^s} & = \p{x_i^s}\p{y_i^s} \\
    \text{Hence, }\quad  \p{x_{i'}^{s'}, y_{i'}^{s'} \mid x_i^s, y_i^s}
 & = \frac{\p{x_{i'}^{s'}, y_{i'}^{s'}}\p{x_{i'}^{s'}, x_i^s}}{\p{x_{i'}^{s'} \mid x_i^s}\p{x_i^s}\p{x_i^s}} \\
 & = \frac{\p{x_{i'}^{s'}, y_{i'}^{s'}}}{\p{x_i^s}}.
  \end{align*}

  Thus, $x_{i'}^{s'}, y_{i'}^{s'} \bot y_i^s \mid x_i^s$. Moreover, Fisher's exact test assumes that $|\{i: x_i^s = 1\}| = n_s$ for any dataset $D = \{t_i, y_i\} _{i=1}^{n_D} \sim \mathcal{D}$. We have
  \begin{align*}
 &   & \{x_{i'}^{s'}, y_{i'}^{s'}\} _{i \in I_s} \bot \{x_i^s, y_i^s\} _{i \in I_s} \mid n_s, n_{s'}, n_D \\
 & \implies & a_s \bot a_{s'} \mid n_s, n_{s'}, n_D                                                        \\
 & \implies & \{p_h\} _{h \in T} \bot \{p_h\} _{h \in F} \mid n_s, n_{s'}, n_D.
  \end{align*}

  This concludes our proof.
\end{proof}

\mysubsection{Proof of Lemma\ref{lem:sigmaincrease}}
\begin{proof}
    First, we remind that $\sigma_t$ and $\sigma_{t+1}$ are obtained as follows
    \begin{align*}
      \sigma_t     & = \max \{\sigma: (\sigma - p_{t-1}^r)|\kappa_t(\sigma)| \leq \delta_t\} \\
      \sigma_{t+1} & = \max \{\sigma: (\sigma - p_t^r)|\kappa_{t+1}(\sigma)| \leq \delta_{t+1}\}.
    \end{align*}

    Because $H_{t+1} \subset H_t$ and $|\kappa_t(\sigma_t)| \leq \frac{\delta_t}{\sigma_t - p_{t-1}^r}$, we have
    \begin{align*}
                 & |\kappa_{t+1}(\sigma_t)| < |\kappa_t(\sigma_t)| \leq \frac{\delta_t}{\sigma_t - p_{t-1}^r} \\
      \implies \quad & |\kappa_{t+1}(\sigma_t)|(\sigma_t - p_t^r) < \delta_t \frac{\sigma_t - p_t^r}{\sigma_t - p_{t-1}^r}.
    \end{align*}

    On the other hand,
    \begin{align*}
      \delta_{t+1} = \delta_t - \tau_t(p_t^r - p_{t-1}^r) + p_t^r = \delta_t \frac{\sigma_t - p_t^r}{\sigma_t - p_{t-1}^r} + p_t^r.
    \end{align*}

    Thus, $|\kappa_{t+1}(\sigma_t)|(\sigma_t - p_t^r) < \delta_{t+1} - p_t^r \leq \delta_{t+1}$ and  $\sigma_{t+1} \geq \sigma_t$ due to the maximal operation.
\end{proof}

\mysubsection{Proof of Theorem~\ref{theorem:lamp}}
\begin{proof}
  Let $R$ be the rejected set by SPUR.\@ Moreover, let $\sigma_\mathrm{Tarone}$ and $R_\mathrm{Tarone}$ be the rejection threshold and the rejected set by T-Bonferroni, respectively. We have that
  \begin{align*}
    \sigma_\mathrm{Tarone}   & = \max\{\sigma \mid \sigma\kappa(\sigma) \leq \alpha\} \\
    R_{\mathrm{Tarone}} & = \{p_h \leq \sigma_\mathrm{Tarone}\} _{h \in H}.
  \end{align*}

  Letting $h \in R_\mathrm{Tarone}^c = H \setminus R_\mathrm{Tarone}$, the following holds for a hypothesis $h^* \in dom(R_\mathrm{Tarone})$:
  \begin{align*}
                 & \nexists_{h \in R_\mathrm{Tarone}}, h \succeq h^*         & \quad (\text{since } h^* \in dom(R_\mathrm{Tarone})) \\
    \text{ and } & \nexists_{h \in R_\mathrm{Tarone}^c}, p_h < p_{h^*} & \quad (\text{since } h^* \in R_\mathrm{Tarone}).
  \end{align*}

  Since $R_\mathrm{Tarone}^c \cup R_\mathrm{Tarone} = H$, $\nexists_{h \in H} h \succeq h^* \land p_h < p_{h^*}$

  On the other hand, $\exists_{h \in H} h \succeq h^* \land p_h < p_{h^*}$ is the necessary condition for the hypothesis $h^*$ to be ignored in the procedure SPUR.\@ Thus, no hypotheses in $dom(R_\mathrm{Tarone})$ will be ignored by SPUR.\@ Moreover, we have $\sigma_1 = \sigma_\mathrm{Tarone}$ since $H_1 = H$ and $p_0 = 0$ and $\sigma_k \geq \sigma_1 \geq \sigma_\mathrm{Tarone}$ for any iteration $k$ by Lemma~\ref{lem:sigmaincrease}. Thus, $p_k^* \leq \sigma_\mathrm{Tarone} \leq \sigma_k$ and it implies that $dom(R_\mathrm{Tarone}) \subseteq R$. This concludes our proof.
\end{proof}

%% file: ch-appendicies/implementation.tex
\begin{table}[h]
  \centering
  \caption{Adopted variables for each dataset and the number of categorizing levels for each variable (in parentheses).}\label{tab:variable-description}
  \vspace{-0.5em}
  \resizebox{1\linewidth}{!}{%
    \begin{tabular}{llll}
      \toprule
                         & family-vars       & utility-vars          & target-class                  \\
      \midrule \midrule
      Crash\cite{data-crash} & street-type ($5$) & traffic-control ($2$) & crash-type ($2$)              \\
                         & time ($24$)       & speed-limit ($7$)     & $y=1$ (injured)               \\
                         &                   & weather ($3$)         &                               \\
                         &                   & light-condition ($3$) &                               \\
      \midrule
      Adult\cite{data-adult} & sex ($2$)         & hours-per-week ($6$)  & income ($2$)                  \\
                         & work-class ($3$)  & education ($10$)      & $y=1$  ($>50$K)               \\
                         & occupation ($15$) &                       &                               \\
      \midrule
      Crime\cite{data-crime} & place ($12$)      & time ($12$)           & crime-type ($3$)              \\
                         & street ($14$)     &                       & $y =$ property/person/society \\
      \bottomrule
    \end{tabular}}
  \vspace{-1em}
\end{table}

Other details and the code of our experiment are available at https://github.com/dizzyvn/SPUR/

%% file: kdd.bbl

\begin{thebibliography}{24}


\ifx \showCODEN    \undefined \def \showCODEN     #1{\unskip}     \fi
\ifx \showDOI      \undefined \def \showDOI       #1{#1}\fi
\ifx \showISBNx    \undefined \def \showISBNx     #1{\unskip}     \fi
\ifx \showISBNxiii \undefined \def \showISBNxiii  #1{\unskip}     \fi
\ifx \showISSN     \undefined \def \showISSN      #1{\unskip}     \fi
\ifx \showLCCN     \undefined \def \showLCCN      #1{\unskip}     \fi
\ifx \shownote     \undefined \def \shownote      #1{#1}          \fi
\ifx \showarticletitle \undefined \def \showarticletitle #1{#1}   \fi
\ifx \showURL      \undefined \def \showURL       {\relax}        \fi
\providecommand\bibfield[2]{#2}
\providecommand\bibinfo[2]{#2}
\providecommand\natexlab[1]{#1}
\providecommand\showeprint[2][]{arXiv:#2}

\bibitem[\protect\citeauthoryear{Agrawal, Srikant, et~al\mbox{.}}{Agrawal
  et~al\mbox{.}}{1994}]%
        {agrawal1994fast}
\bibfield{author}{\bibinfo{person}{Rakesh Agrawal},
  \bibinfo{person}{Ramakrishnan Srikant}, {et~al\mbox{.}}}
  \bibinfo{year}{1994}\natexlab{}.
\newblock \showarticletitle{Fast algorithms for mining association rules}. In
  \bibinfo{booktitle}{\emph{Proc. 20th int. conf. very large data bases,
  VLDB}}, Vol.~\bibinfo{volume}{1215}. \bibinfo{pages}{487--499}.
\newblock


\bibitem[\protect\citeauthoryear{Benjamini and Hochberg}{Benjamini and
  Hochberg}{2000}]%
        {benjamini2000adaptive}
\bibfield{author}{\bibinfo{person}{Yoav Benjamini} {and} \bibinfo{person}{Yosef
  Hochberg}.} \bibinfo{year}{2000}\natexlab{}.
\newblock \showarticletitle{On the adaptive control of the false discovery rate
  in multiple testing with independent statistics}.
\newblock \bibinfo{journal}{\emph{Journal of educational and Behavioral
  Statistics}} \bibinfo{volume}{25}, \bibinfo{number}{1}
  (\bibinfo{year}{2000}), \bibinfo{pages}{60--83}.
\newblock


\bibitem[\protect\citeauthoryear{Fisher et~al\mbox{.}}{Fisher
  et~al\mbox{.}}{1950}]%
        {fisher1950statistical}
\bibfield{author}{\bibinfo{person}{Ronald~Aylmer Fisher} {et~al\mbox{.}}}
  \bibinfo{year}{1950}\natexlab{}.
\newblock \showarticletitle{Statistical methods for research workers.}
\newblock \bibinfo{journal}{\emph{Statistical methods for research workers.}}
  \bibinfo{number}{llth ed. revised} (\bibinfo{year}{1950}).
\newblock


\bibitem[\protect\citeauthoryear{Gan, Lin, Fournier-Viger, Chao, Tseng, and
  Yu}{Gan et~al\mbox{.}}{2018}]%
        {gan2018survey}
\bibfield{author}{\bibinfo{person}{Wensheng Gan}, \bibinfo{person}{Jerry
  Chun-Wei Lin}, \bibinfo{person}{Philippe Fournier-Viger},
  \bibinfo{person}{Han-Chieh Chao}, \bibinfo{person}{Vincent~S Tseng}, {and}
  \bibinfo{person}{Philip~S Yu}.} \bibinfo{year}{2018}\natexlab{}.
\newblock \showarticletitle{A survey of utility-oriented pattern mining}.
\newblock \bibinfo{journal}{\emph{arXiv preprint arXiv:1805.10511}}
  (\bibinfo{year}{2018}).
\newblock


\bibitem[\protect\citeauthoryear{H{\"a}m{\"a}l{\"a}inen}{H{\"a}m{\"a}l{\"a}inen}{2012}]%
        {hamalainen2012kingfisher}
\bibfield{author}{\bibinfo{person}{Wilhelmiina H{\"a}m{\"a}l{\"a}inen}.}
  \bibinfo{year}{2012}\natexlab{}.
\newblock \showarticletitle{Kingfisher: an efficient algorithm for searching
  for both positive and negative dependency rules with statistical significance
  measures}.
\newblock \bibinfo{journal}{\emph{Knowledge and information systems}}
  \bibinfo{volume}{32}, \bibinfo{number}{2} (\bibinfo{year}{2012}),
  \bibinfo{pages}{383--414}.
\newblock


\bibitem[\protect\citeauthoryear{Holm}{Holm}{1979}]%
        {holm1979simple}
\bibfield{author}{\bibinfo{person}{Sture Holm}.}
  \bibinfo{year}{1979}\natexlab{}.
\newblock \showarticletitle{A simple sequentially rejective multiple test
  procedure}.
\newblock \bibinfo{journal}{\emph{Scandinavian journal of statistics}}
  (\bibinfo{year}{1979}), \bibinfo{pages}{65--70}.
\newblock


\bibitem[\protect\citeauthoryear{Kohavi and barry Becker}{Kohavi and barry
  Becker}{1996}]%
        {data-adult}
\bibfield{author}{\bibinfo{person}{Ronny Kohavi} {and} \bibinfo{person}{barry
  Becker}.} \bibinfo{year}{1996}\natexlab{}.
\newblock \bibinfo{title}{Adult Data Set}.
\newblock
\newblock


\bibitem[\protect\citeauthoryear{Komiyama, Ishihata, Arimura, Nishibayashi, and
  Minato}{Komiyama et~al\mbox{.}}{2017}]%
        {komiyama2017statistical}
\bibfield{author}{\bibinfo{person}{Junpei Komiyama}, \bibinfo{person}{Masakazu
  Ishihata}, \bibinfo{person}{Hiroki Arimura}, \bibinfo{person}{Takashi
  Nishibayashi}, {and} \bibinfo{person}{Shin-ichi Minato}.}
  \bibinfo{year}{2017}\natexlab{}.
\newblock \showarticletitle{Statistical emerging pattern mining with multiple
  testing correction}. In \bibinfo{booktitle}{\emph{Proceedings of the 23rd ACM
  SIGKDD International Conference on Knowledge Discovery and Data Mining}}.
  ACM, \bibinfo{pages}{897--906}.
\newblock


\bibitem[\protect\citeauthoryear{Lee, Park, and Moon}{Lee
  et~al\mbox{.}}{2013}]%
        {lee2013utility}
\bibfield{author}{\bibinfo{person}{Dongwon Lee}, \bibinfo{person}{Sung-Hyuk
  Park}, {and} \bibinfo{person}{Songchun Moon}.}
  \bibinfo{year}{2013}\natexlab{}.
\newblock \showarticletitle{Utility-based association rule mining: A marketing
  solution for cross-selling}.
\newblock \bibinfo{journal}{\emph{Expert Systems with applications}}
  \bibinfo{volume}{40}, \bibinfo{number}{7} (\bibinfo{year}{2013}),
  \bibinfo{pages}{2715--2725}.
\newblock


\bibitem[\protect\citeauthoryear{Llinares-L{\'o}pez, Papaxanthos, Bodenham,
  Roqueiro, Investigators, and Borgwardt}{Llinares-L{\'o}pez
  et~al\mbox{.}}{2017}]%
        {llinares2017genome}
\bibfield{author}{\bibinfo{person}{Felipe Llinares-L{\'o}pez},
  \bibinfo{person}{Laetitia Papaxanthos}, \bibinfo{person}{Dean Bodenham},
  \bibinfo{person}{Damian Roqueiro}, \bibinfo{person}{COPDGene Investigators},
  {and} \bibinfo{person}{Karsten Borgwardt}.} \bibinfo{year}{2017}\natexlab{}.
\newblock \showarticletitle{Genome-wide genetic heterogeneity discovery with
  categorical covariates}.
\newblock \bibinfo{journal}{\emph{Bioinformatics}} \bibinfo{volume}{33},
  \bibinfo{number}{12} (\bibinfo{year}{2017}), \bibinfo{pages}{1820--1828}.
\newblock


\bibitem[\protect\citeauthoryear{Llinares-L{\'o}pez, Sugiyama, Papaxanthos, and
  Borgwardt}{Llinares-L{\'o}pez et~al\mbox{.}}{2015}]%
        {llinares2015fast}
\bibfield{author}{\bibinfo{person}{Felipe Llinares-L{\'o}pez},
  \bibinfo{person}{Mahito Sugiyama}, \bibinfo{person}{Laetitia Papaxanthos},
  {and} \bibinfo{person}{Karsten Borgwardt}.} \bibinfo{year}{2015}\natexlab{}.
\newblock \showarticletitle{Fast and memory-efficient significant pattern
  mining via permutation testing}. In \bibinfo{booktitle}{\emph{Proceedings of
  the 21th ACM SIGKDD international conference on knowledge discovery and data
  mining}}. ACM, \bibinfo{pages}{725--734}.
\newblock


\bibitem[\protect\citeauthoryear{Minato, Uno, Tsuda, Terada, and Sese}{Minato
  et~al\mbox{.}}{2014}]%
        {minato2014fast}
\bibfield{author}{\bibinfo{person}{Shin-ichi Minato}, \bibinfo{person}{Takeaki
  Uno}, \bibinfo{person}{Koji Tsuda}, \bibinfo{person}{Aika Terada}, {and}
  \bibinfo{person}{Jun Sese}.} \bibinfo{year}{2014}\natexlab{}.
\newblock \showarticletitle{A fast method of statistical assessment for
  combinatorial hypotheses based on frequent itemset enumeration}. In
  \bibinfo{booktitle}{\emph{Joint European Conference on Machine Learning and
  Knowledge Discovery in Databases}}. Springer, \bibinfo{pages}{422--436}.
\newblock


\bibitem[\protect\citeauthoryear{Pellegrina, Riondato, and Vandin}{Pellegrina
  et~al\mbox{.}}{2019}]%
        {pellegrina2019spumante}
\bibfield{author}{\bibinfo{person}{Leonardo Pellegrina},
  \bibinfo{person}{Matteo Riondato}, {and} \bibinfo{person}{Fabio Vandin}.}
  \bibinfo{year}{2019}\natexlab{}.
\newblock \showarticletitle{SPuManTE: Significant Pattern Mining with
  Unconditional Testing}.
\newblock


\bibitem[\protect\citeauthoryear{Rosenthal and Rubin}{Rosenthal and
  Rubin}{1983}]%
        {rosenthal1983ensemble}
\bibfield{author}{\bibinfo{person}{Robert Rosenthal} {and}
  \bibinfo{person}{Donald~B Rubin}.} \bibinfo{year}{1983}\natexlab{}.
\newblock \showarticletitle{Ensemble-adjusted p values.}
\newblock \bibinfo{journal}{\emph{Psychological Bulletin}}
  \bibinfo{volume}{94}, \bibinfo{number}{3} (\bibinfo{year}{1983}),
  \bibinfo{pages}{540}.
\newblock


\bibitem[\protect\citeauthoryear{Service}{Service}{2015}]%
        {data-crime}
\bibfield{author}{\bibinfo{person}{MCG~ESB Service}.}
  \bibinfo{year}{2015}\natexlab{}.
\newblock \bibinfo{title}{Crime}.
\newblock
\newblock


\bibitem[\protect\citeauthoryear{Service}{Service}{2017}]%
        {data-crash}
\bibfield{author}{\bibinfo{person}{MCG~ESB Service}.}
  \bibinfo{year}{2017}\natexlab{}.
\newblock \bibinfo{title}{Crash Reporting - Drivers Data}.
\newblock
\newblock


\bibitem[\protect\citeauthoryear{Shie, Philip, and Tseng}{Shie
  et~al\mbox{.}}{2012}]%
        {shie2012efficient}
\bibfield{author}{\bibinfo{person}{Bai-En Shie}, \bibinfo{person}{S~Yu Philip},
  {and} \bibinfo{person}{Vincent~S Tseng}.} \bibinfo{year}{2012}\natexlab{}.
\newblock \showarticletitle{Efficient algorithms for mining maximal high
  utility itemsets from data streams with different models}.
\newblock \bibinfo{journal}{\emph{Expert Systems with Applications}}
  \bibinfo{volume}{39}, \bibinfo{number}{17} (\bibinfo{year}{2012}),
  \bibinfo{pages}{12947--12960}.
\newblock


\bibitem[\protect\citeauthoryear{Sugiyama and Borgwardt}{Sugiyama and
  Borgwardt}{2019}]%
        {sugiyama2019finding}
\bibfield{author}{\bibinfo{person}{Mahito Sugiyama} {and}
  \bibinfo{person}{Karsten Borgwardt}.} \bibinfo{year}{2019}\natexlab{}.
\newblock \showarticletitle{Finding statistically significant interactions
  between continuous features}. In \bibinfo{booktitle}{\emph{Proceedings of the
  28th International Joint Conference on Artificial Intelligence}}. AAAI Press,
  \bibinfo{pages}{3490--3498}.
\newblock


\bibitem[\protect\citeauthoryear{Tarone}{Tarone}{1990}]%
        {tarone1990modified}
\bibfield{author}{\bibinfo{person}{Robert~E Tarone}.}
  \bibinfo{year}{1990}\natexlab{}.
\newblock \showarticletitle{A modified Bonferroni method for discrete data}.
\newblock \bibinfo{journal}{\emph{Biometrics}} (\bibinfo{year}{1990}),
  \bibinfo{pages}{515--522}.
\newblock


\bibitem[\protect\citeauthoryear{Terada, Okada-Hatakeyama, Tsuda, and
  Sese}{Terada et~al\mbox{.}}{2013}]%
        {terada2013statistical}
\bibfield{author}{\bibinfo{person}{Aika Terada}, \bibinfo{person}{Mariko
  Okada-Hatakeyama}, \bibinfo{person}{Koji Tsuda}, {and} \bibinfo{person}{Jun
  Sese}.} \bibinfo{year}{2013}\natexlab{}.
\newblock \showarticletitle{Statistical significance of combinatorial
  regulations}.
\newblock \bibinfo{journal}{\emph{Proceedings of the National Academy of
  Sciences}} \bibinfo{volume}{110}, \bibinfo{number}{32}
  (\bibinfo{year}{2013}), \bibinfo{pages}{12996--13001}.
\newblock


\bibitem[\protect\citeauthoryear{Terada, Tsuda, et~al\mbox{.}}{Terada
  et~al\mbox{.}}{2016}]%
        {terada2016significant}
\bibfield{author}{\bibinfo{person}{Aika Terada}, \bibinfo{person}{Koji Tsuda},
  {et~al\mbox{.}}} \bibinfo{year}{2016}\natexlab{}.
\newblock \showarticletitle{Significant pattern mining with confounding
  variables}. In \bibinfo{booktitle}{\emph{Pacific-Asia Conference on Knowledge
  Discovery and Data Mining}}. Springer, \bibinfo{pages}{277--289}.
\newblock


\bibitem[\protect\citeauthoryear{Webb}{Webb}{2007}]%
        {webb2007discovering}
\bibfield{author}{\bibinfo{person}{Geoffrey~I Webb}.}
  \bibinfo{year}{2007}\natexlab{}.
\newblock \showarticletitle{Discovering significant patterns}.
\newblock \bibinfo{journal}{\emph{Machine learning}} \bibinfo{volume}{68},
  \bibinfo{number}{1} (\bibinfo{year}{2007}), \bibinfo{pages}{1--33}.
\newblock


\bibitem[\protect\citeauthoryear{Webb and Petitjean}{Webb and
  Petitjean}{2016}]%
        {webb2016multiple}
\bibfield{author}{\bibinfo{person}{Geoffrey~I Webb} {and}
  \bibinfo{person}{Fran{\c{c}}ois Petitjean}.} \bibinfo{year}{2016}\natexlab{}.
\newblock \showarticletitle{A multiple test correction for streams and cascades
  of statistical hypothesis tests}. In \bibinfo{booktitle}{\emph{Proceedings of
  the 22nd ACM SIGKDD International Conference on Knowledge Discovery and Data
  Mining}}. ACM, \bibinfo{pages}{1255--1264}.
\newblock


\bibitem[\protect\citeauthoryear{Yao and Hamilton}{Yao and Hamilton}{2006}]%
        {yao2006mining}
\bibfield{author}{\bibinfo{person}{Hong Yao} {and} \bibinfo{person}{Howard~J
  Hamilton}.} \bibinfo{year}{2006}\natexlab{}.
\newblock \showarticletitle{Mining itemset utilities from transaction
  databases}.
\newblock \bibinfo{journal}{\emph{Data \& Knowledge Engineering}}
  \bibinfo{volume}{59}, \bibinfo{number}{3} (\bibinfo{year}{2006}).
\newblock


\end{thebibliography}
